\documentclass[journal]{IEEEtran}

\ifCLASSINFOpdf
\else
\fi
\usepackage[cmex10]{amsmath}

\usepackage{multirow}
\usepackage{pstricks, pst-plot, pst-grad, pstricks-add, pst-node, pst-3d, pstricks-add}
\usepackage{epstopdf}
\usepackage{algorithm}
\usepackage{algpseudocode}

\usepackage{graphicx}
\usepackage{amssymb}
\usepackage{bm}
\usepackage{slashbox}
\usepackage{amsthm}
\usepackage{array}
\usepackage{enumerate}

\DeclareMathAlphabet      {\mathbfit}{OML}{cmm}{b}{it}
\DeclareMathAlphabet      {\mathbfrm}{OT1}{cmr}{b}{n}

\theoremstyle{plain}
\newtheorem{thm}{Theorem}
\newtheorem{lem}{Lemma}
\newtheorem{prop}{Proposition}

\theoremstyle{definition}
\newtheorem{defn}{Definition}

\newtheorem{assm}{Assumption}

\theoremstyle{remark}
\newtheorem{rem}{Remark}

\newcommand{\slfrac}[2]{\left.#1\middle/#2\right.}

\begin{document}
%
\title{Distributed Asynchronous Optimization Framework for the MISO Interference Channel}
%
%
%

\author{
\IEEEauthorblockN{Stefan Wesemann and Gerhard Fettweis}\\
\IEEEauthorblockA{Vodafone Chair Mobile Communications Systems, TU-Dresden, Germany\\
Email: \{stefan.wesemann, fettweis\}@ifn.et.tu-dresden.de}}

\author{Stefan~Wesemann,~\IEEEmembership{Student~Member,~IEEE,}
        and~Gerhard~Fettweis,~\IEEEmembership{Fellow,~IEEE}
\thanks{\copyright \copyright 2014 IEEE. Personal use of this material is permitted. Permission
from IEEE must be obtained for all other uses, in any current or future media,
including reprinting/republishing this material for advertising or promotional
purposes, creating new collective works, for resale or redistribution to servers
or lists, or reuse of any copyrighted component of this work in other works.}%
\thanks{S. Wesemann and G. Fettweis are with the Vodafone Chair Mobile Communications Systems, TU-Dresden, Germany, e-mail: \{stefan.wesemann, fettweis\}@ifn.et.tu-dresden.de.} 
}

\maketitle

\begin{abstract}
We study the distributed optimization of transmit strategies in a multiple-input, single-output (MISO) interference channel (IFC). Existing distributed algorithms rely on strictly synchronized update steps by the individual users. They require a global synchronization mechanism and potentially suffer from the synchronization penalty caused by e.g., backhaul communication delays and fixed update sequences.
We establish a general optimization framework that allows asynchronous update steps. The users perform their computations at arbitrary instants of time, and do not wait for information that has been sent to them. Based on certain bounds on the amount of asynchronism that is present in the execution of the algorithm, we are able to characterize its convergence. As illustrated by our numerical results, the proposed algorithm is not excessively slowed down by neither communication delays, nor by specific update orders, and thus enables faster convergence to (local) optimal solution.
\end{abstract}


%
\IEEEpeerreviewmaketitle

\section{Introduction}
%
%
%
%

\IEEEPARstart{D}{istributed} 
interference coordination in wireless networks \cite{5371760} is of special interest, since the alternative of centralized control involves added infrastructure, latency and network vulnerability. We consider networks that can be modeled as a set of mutually interfering multiple-input, single-output (MISO) links \cite{1405330}, each representing a user. Although the optimal transmit strategy requires complex signal-level en-/decoding cooperations among the users, we assume that each user employs single-stream beamforming with single-user detection. Our objective is the maximization of the sum of all user utilities, which is referred to as the sum utility problem.

The primary focus of this work is on the design and evaluation of a distributed \emph{asynchronous} optimization framework, in which the users update their transmission strategies autonomously, based on locally available channel state information and the (possibly delayed) exchange of optimization parameters via backhaul. In contrast to synchronous algorithms, the proposed method does not rely on any centralized control, and can cope with outdated information; that is, the local update computations never wait for inputs but keep performing whatever information is currently available.\\
The crucial question is whether or not asynchronism helps to alleviate the synchronization penalty \cite{Bertsekas:1989:PDC} that is caused by specific update orders, backhaul delays and differences in the computation intervals.

The considered MISO interference channel (IFC) is a well-investigated model (see the excellent tutorial in \cite{CIT-069}). However, its distributed asynchronous optimization has still been an open problem. We now give a brief reference to work that relates to the sum utility problem (SUP) in the MISO IFC.

\vspace{-0.2cm}
\subsection{Related Work}
Determining the sum utility \emph{optimal} transmit strategy is proven to be NP-hard in general, as shown in \cite{5638157}. 
Interestingly, for some special cases there exist distributed optimal closed-form solutions. Sum-rate optimal solutions are obtained by the maximum ratio transmission (MRT) beamformers \cite{795811} at low signal-to-noise ratios (SNR), and by the zero-forcing (ZF) beamformers at high SNRs, but only in scenarios where zero-forcing is possible for all users.
Under general conditions, the monotonic optimization framework from \cite{Tuy:2000} provides mechanisms for finding an $\epsilon$-optimal solution in a finite number of iterations, but only if the user utility functions satisfy certain monotonicity properties. Examples of such centralized algorithms are found in \cite{5504616},\cite{6129540},\cite{5946278}. Due to the NP-hardness, the number of required iterations scales exponentially with the number of users; that is, attempting to find a real-time optimal solution for a large number of users is infeasible. Thus, the framework is only suitable for computing benchmarks.

There exists a multitude of distributed \emph{synchronous} algorithms with guaranteed convergence to a stationary point of the SUP. First note that the (optimal) closed-form solutions for the sum-rate maximization (SRMax) problem (i.e., MRT and ZF) can be generalized by a minimum mean square error beamforming structure, yielding the maximum virtual signal-to-interference-plus-noise (SINR) beamformer in \cite{VSINR}. However, the virtual SINR maximization always results in full power for all users, which is in general not sum-rate optimal. A closely related algorithm for the weighted SRMax problem is given in \cite{5610967}. The algorithm employs a high-SINR approximation in order to obtain fully decoupled subproblems; that is, each transmitter requires only local CSI to generate a \emph{near} optimal solution.\\
In order to achieve (local) optimal solutions for a broader class of utility functions, one has to resort to iterative algorithms. In \cite{5205801}, a distributed pricing (DP) algorithm for the MISO IFC has been proposed, in which each user iteratively maximizes its own utility function plus the summation of the first-order approximation of all other users' utility functions at the current operating point. The linearization is based on so-called interference prices, which must be exchanged between the users. The monotonic convergence to a stationary point is guaranteed if the (twice differentiable) utility functions are convex with respect to the interference power, and the users perform sequential updates with current knowledge of the interference prices. An extension of the algorithm for the MIMO interference channel is found in \cite{5198927}.\\ 
A closely related cyclic coordinate descent (CCD) algorithm is found in \cite{5638157}. This algorithm also requires sequential update steps and current knowledge of the optimization parameters. By assuming that every user has a-priori knowledge of all (twice continuously differentiable) utility functions, it is sufficient for the users to announce the numerator and the denominator of the SINR after each iteration. However, as shown in \cite{schmidtDiss}, the number of iterations required for convergence is very large, especially at high SNR, because the algorithm does not make any assumptions on the curvature of the utility functions.\\
In \cite{5470055}, a weighted sum mean-square error (MSE) minimization is proposed, in which the weights are adaptively chosen to mimic the behavior of arbitrary utility functions. Each time the weights are updated and communicated among the users, the proposed algorithm alternates between the updates of the receiver gains and the transmit beamformers. If the user utility functions are convex in the MSE then the solution monotonically converges to a stationary point. An extension for the MIMO interfering broadcast channel is found in \cite{5756489}.\\
Note that the described algorithms can not cope with outdated information and thus rely on a synchronization mechanism, which introduces idle periods. 

By focusing on distributed \emph{asynchronous} approaches, our literature study identified only one algorithm for general user utility functions. In \cite{4797605} an asynchronous distributed pricing (ADP) algorithm is proposed for the \emph{two-user} MISO IFC, in which the users perform their update steps at arbitrary instants of time, based on possibly outdated information. By re-parameterizing the original problem, the authors show that the algorithm corresponds to best response updates in a supermodular game, which relies on the principle of strategic complements (i.e., the strategies of the two users mutually reinforce one another). If certain beamformer initializations are used and the utility functions satisfy some special criteria of the coefficient of relative risk aversion, then the solution of the ADP algorithm converges monotonically to a stationary point of the sum utility problem. However, its convergence can only be established for the two-user case.

\subsection{Contributions}
We adopt the distributed computation model from \cite{1104412} in order to formulate an asynchronous optimization algorithm for the general MISO IFC. We start with the re-parametrization of the SUP in terms of received signal powers (so-called power gains), which entails the following advantages:
\begin{itemize}
	\item Typically, the user utility functions (e.g., SINR, data rate) are defined in terms of signal and interference powers. Any phase rotation of the received signal is irrelevant. Consequently, the power gain based problem representation reflects the essential problem structure and provides a reduced parameter space.
	\item By focusing on distributed optimization approaches, the coupling between the subproblems can be efficiently described by few real-valued scalars, which stands in contrast to multi-dimensional complex matrices that arise in the original beamforming domain.
	\item The underlying power gain regions, which serve as the constraint set, admit viewpoints from convex geometry \cite{hiriart2001fundamentals} for the characterization of (local) optimal operating points.
\end{itemize}
Note that the feasible set of the re-parameterized problem is non-convex, when focusing on single-stream beamforming only. This may appear as a disadvantage of the re-parameterization, since most of the optimization approaches rely on convex constraint sets. However, we can convexify our constraint sets by employing a rank-relaxation for the underlying transmit correlation matrices, and we show that this relaxation is tight for all stationary points of the SUP.\\
This enables the application of the distributed scaled gradient projection (SGP) algorithm \cite{Bertsekas:1989:PDC}, which provides (local) optimal solutions of the relaxed SUP. The underlying projection onto the power gain region is formulated as a (convex) quadratic semi-definite program \cite{QSDP0}. Moreover, we show how to extract the corresponding beamforming vectors by solving an interference-constrained beamforming problem.\\
Inspired by \cite[Section 5.6]{2140467}, we formulate explicit bounds on the backhaul delays and curvatures of the sum utility function, in order to provide sufficient conditions that ensure the convergence of the asynchronous SGP algorithm to a stationary point of the SUP. 
Finally, we investigate the convergence rate of different synchronous and asynchronous algorithms by means of numerical experiments.\\

\emph{Outline}: In Section II, we provide the system model and introduce the power gain region. In Section III, we formulate the sum utility problem and provide necessary optimality conditions. In Section IV, we describe the distributed asynchronous optimization framework and adopt the scaled gradient projection method. In Section V, we provide simulation results before we conclude in Section VI.\\

\emph{Notation}: Vectors and matrices are written in lowercase and uppercase boldface letters, respectively. The notation $x_{k,l}$ describes the $l$-th component of the vector $\mathbfit{x}_k$. The Euclidean norm of a vector $\mathbfit{a}\in\mathbb{C}^N$, is written as $\left\|\mathbfit{a}\right\|$. $(\cdot)^T$ and $(\cdot)^H$ denote the transpose and Hermitian transpose, respectively. Let $\lambda_1(\mathbfit{A}) \geq \ldots \geq \lambda_N(\mathbfit{A})$ be the eigenvalues of the matrix $\mathbfit{A}\in\mathbb{C}^{N\times N}$, and $\mathcal{E}_k(\mathbfit{A}),1\leq k \leq N$ are the corresponding eigenspaces. The dominant eigenvector of the matrix $\mathbfit{Z}$ is denoted by $v_{max}(\mathbfit{Z})$. $\mathbfit{Z}\succeq 0$ means that $\mathbfit{Z}$ is positive semi-definite. The rank and trace of a matrix $\mathbfit{Z}$ are given by $\mathrm{rank}(\mathbfit{Z})$ and $\mathrm{tr}(\mathbfit{Z})$, respectively. $\Re(x)$ and $\Im(x)$ denote the real and imaginary parts of $x$. We use $\mathbb{R}_+$ (resp. $\mathbb{R}_{++}$) to denote the set of nonnegative (resp. positive) real numbers.

\section{System Model and Power Gain Region}\label{SysModProbForm}
\subsection{System Model}
We consider a narrowband, time-invariant MISO interference channel with $K$ users. Each user consists of a transmitter/receiver pair, where the transmitter has $N$ antennas and the receiver is assumed to have a single effective antenna. As illustrated in Fig. \ref{fig:chanKnow}, the $k$-th receiver observes a superposition of signals from all transmitters but is interested only in the transmit signal from its associated transmitter. The received symbol at receiver $k$ is given by
\begin{align}
r_k = \sum_{l=1}^K \mathbfit{h}_{lk}^H\mathbfit{s}_l + n_k,
\end{align}
where $\mathbfit{s}_l \in \mathbb{C}^{N}$ denotes the transmit signal of the $l$-th transmitter; $\mathbfit{h}_{lk}\in \mathbb{C}^{N}$ denotes the channel vector between the $l$-th transmitter and the $k$-th receiver. Each receiver experiences additive noise $n_k$ with zero mean and variance $\sigma^2$. The stochastic transmit signals are modeled as zero-mean r.v. with signal correlation matrices $\mathbfit{Q}_k=\mathbb{E}\left\{\mathbfit{s}_k\mathbfit{s}_k^H\right\} \in \mathbb{C}^{N\times N}$. In case of multi-stream beamforming, we have $\mathrm{rank}(\mathbfit{Q}_k)>1$; that is, the individual data streams are assumed to be statistically independent. Each transmitter $k$ has a total power constraint, given by $\mathrm{tr}(\mathbfit{Q}_k) \leq 1$. At receiver side, each receiver treats the co-channel interference as additional noise.

We argue that it may not be reasonable to assume that all the channel state information (CSI) is shared by all users. 
\begin{assm}[Local CSI Knowledge] Each user $k$ has only local CSI; that is,
\begin{itemize}
	\item it knows perfectly the channel vector $\mathbfit{h}_{kl}$ between its transmitter $k$ and \emph{each} receiver $l$,
	\item it knows perfectly the scalar channel gain $\left\|\mathbfit{h}_{lk}\right\|^2$ between \emph{each} transmitter $l$ and its receiver $k$.
\end{itemize}
\end{assm}
The local CSI of the $k$-th user is illustrated in Figure \ref{fig:chanKnow}. It can be obtained by using uplink pilots (see, e.g., \cite{5722348}) in time-division-duplex systems or through feedback from receivers (see, e.g., \cite{5466522}) in frequency-division-duplex systems. Note that the channel gain information $\left\|\mathbfit{h}_{lk}\right\|^2$ is only needed for the convergence speed-ups described in Section \ref{convSpeedUp}. 

\begin{figure}[t]
\centering
\includegraphics[width=8cm]{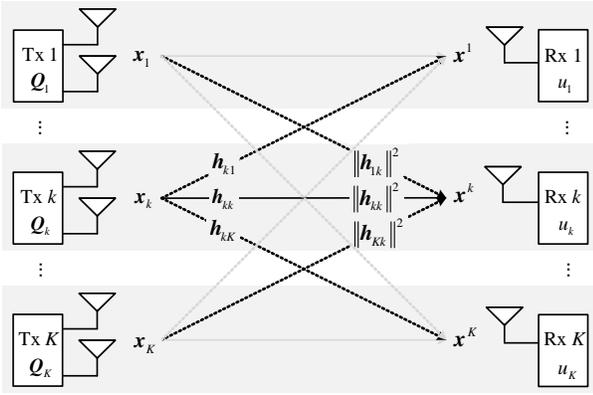}
\caption{$K$-user MISO IFC, illustrated for $N=2$ transmit antennas. Required channel knowledge for the $k$-th user is marked black.}
\label{fig:chanKnow}
\end{figure}


\subsection{Concept of the Power Gain Region}
By the nature of the interference channel, each transmitted signal will in general affect all users. Here, we characterize the impact of each transmitter by its power gain vector, which allows an efficient description of the interactions between a transmitter and all receivers. Consider a transmit signal of the $k$-th transmitter with correlation matrix $\mathbfit{Q}_k$. The received signal power at user $l$ is given by the power gain $x_{k,l}\left(\mathbfit{Q}_k\right) = \mathbfit{h}_{kl}^H \mathbfit{Q}_k \mathbfit{h}_{kl}$. The $K$-tuple of simultaneously achievable power gains from transmitter $k$ forms the transmit power gain vector $\mathbfit{x}_k\left(\mathbfit{Q}_k\right)=[x_{k,1}\left(\mathbfit{Q}_k\right),\ldots,x_{k,K}\left(\mathbfit{Q}_k\right)]^T$. For ease of notation, we introduce the power gain matrix $\mathbfit{X}\in\mathbb{R}_+^{K\times K}$ that collects all power gains of the network, given by
\begin{align}
\mathbfit{X}\left(\mathbfit{Q}_1,\ldots,\mathbfit{Q}_K\right) =  \left[\mathbfit{x}_1\left(\mathbfit{Q}_1\right),\ldots,\mathbfit{x}_K\left(\mathbfit{Q}_K\right)\right]
.\label{eq_pgmatrix}
\end{align}
Note that the $l$-th row of matrix $\mathbfit{X}$ represents the receive power gain vector $\mathbfit{x}^l\left(\mathbfit{Q}_1,\ldots,\mathbfit{Q}_K\right)\in \mathbb{R}_+^{1\times K}$, which contains the power gains that are experienced by the $l$-th receiver.\\
Next, we define the set of feasible transmit power gain vectors for the $k$-th transmitter.
\begin{defn}[Power Gain Region]
 The power gain region $\Omega_k\subset \mathbb{R}_+^{K}$ of the $k$-th transmitter is defined as the set of all achievable power gain vectors $\mathbfit{x}_k\left(\mathbfit{Q}_k\right)$, and is given by
\begin{align}
\Omega_k=\left\{\mathbfit{x}_k\left(\mathbfit{Q}_k\right):  \mathbfit{Q}_k \in \mathcal{Q} \right\}, \label{eq:pg_region:def}
\end{align}
where $\mathcal{Q}=\left\{\mathbfit{Q} \in \mathbb{C}^{N \times N}, \mathrm{tr}(\mathbfit{Q})\leq 1, \mathbfit{Q}\succeq 0\right\}$.
\end{defn}
The power gain region was originally introduced in \cite{5643183}, and is called channel gain region in \cite{CIT-069}. By \cite[Lemma 1]{5643183}, the set $\Omega_k$ is compact and convex.

\begin{rem}
The definition of the power gain region utilizes transmit correlation matrices of arbitrary rank. If we restrict the correlation matrices to be rank one (i.e., correlation matrices that correspond to single-stream beamforming) then the resulting feasible set of power gain vectors is not necessarily convex (see Appendix \ref{app_review}).
\end{rem}

\section{Performance Measures and Optimality Conditions} \label{perfOptCond}
In this section, we seek to characterize the performance of the wireless network by means of utility functions. Therefore, we split the utility measure into two parts: (1) the user utility that is achieved by each user; and (2) the system utility which induces an order on the vectors of simultaneously achievable user utilities. 

\subsection{User Utilities, Utility Region and Pareto Optimality}\label{UserUtilParetoOpt}
We start with the definition of the user utilities and the characterization of efficient operating points. The performance of the $k$-th user is measured by the utility $u_k:\mathbb{R}_+^{1\times K} \rightarrow \mathbb{R}_+$, which is a function of the receive power gain vector $\mathbfit{x}^k\left(\mathbfit{Q}_1,\ldots,\mathbfit{Q}_K\right)$.
\begin{assm}[User Utility Properties]\label{ukProp} The user utility function $u_k(\mathbfit{x}^k\left(\mathbfit{Q}_1,\ldots,\mathbfit{Q}_K\right))$ has the following two properties:
\begin{enumerate}
	\item $u_k$ is strictly monotonically increasing in the power gain $x_{k,k}(\mathbfit{Q}_k)$ from its associated transmitter $k$,
	\item $u_k$ is strictly monotonically decreasing in the power gain $x_{l,k}(\mathbfit{Q}_l)$  from transmitter $l\neq k$.
\end{enumerate}
Without loss of generality, we assume $u_k=0$ if and only if $x_{k,k}=0$.
\end{assm}
Typical examples on user utility functions are the signal-to-interference-plus-noise ratio (SINR), the achievable information rate, and the bit error rate.

Each vector $[u_1,\ldots,u_K]^T$ of simultaneously achievable user utilities represents a feasible operating point. The set of all achievable operating points constitutes the utility region $\mathcal{U}\subset\mathbb{R}_+^K$, defined as
\begin{align}
\mathcal{U}:&=\left\{\left[u_1(\mathbfit{x}^1),\ldots,u_K(\mathbfit{x}^K)\right]^T:\mathbfit{Q}_k\in\mathcal{Q},\forall k\right\}.
\end{align}
Note that there is no total order of the utility vectors in $\mathcal{U}$. However, we can find efficient operating points in $\mathcal{U}$ which are preferable because they are not dominated by any other feasible point. These  points are called Pareto optimal and have the characteristic property that it is impossible to improve the utility of one user without simultaneously degrading the utility of at least one other user.
\begin{defn}[Pareto Optimality]
A point $\mathbfit{u}\in\mathcal{U}$ is Pareto optimal if there is no other tuple $\mathbfit{u}'\in\mathcal{U}$ such that $\mathbfit{u}' \geq \mathbfit{u}$, where the inequality is component-wise and strict for at least one component. The set of all Pareto optimal operating points constitutes the \emph{Pareto boundary} $\mathcal{PB}(\mathcal{U})$.
\end{defn}
In \cite{5426006}, \cite{5504193} it is shown that single-stream beamforming (i.e. signal correlation matrices $\mathbfit{Q}_k$ with $\mathrm{rank}(\mathbfit{Q}_k)\leq 1$) is sufficient for achieving all Pareto optimal points. An alternative proof based on the power gain region is made in \cite{CIT-069}, \cite{5643183}. However, the underlying proof turned out to be incomplete as illustrated in Appendix \ref{prfParOpt}.
\begin{thm}[Sufficiency of Single-Stream Beamforming for Pareto Optimality] \label{thm1}
All Pareto optimal points in the utility region $\mathcal{U}$ can be achieved using single-stream beamforming.
\end{thm}
\begin{proof}
The proof, completing the earlier arguments, is provided in Appendix \ref{prfParOpt}.
\end{proof}

\subsection{Sum Utility Problem and Optimality Conditions} \label{SupOptCond}
By introducing a system utility function $U:\mathcal{U}\rightarrow\mathbb{R}$, we impose a subjective order on the elements in $\mathcal{U}$. Herein, we focus on the dependency of $U(u_1,\ldots,u_K)$ with respect to the power gains $\mathbfit{X}\in\Omega_1\times\ldots\times\Omega_K$. For brevity we write $U(\mathbfit{X})$ instead of the function composition $\left(U\circ(u_1,\ldots,u_K)\right)(\mathbfit{X})$.

\begin{assm}[System Utility Properties]\label{Uprop}
The system utility function $U(\mathbfit{X})$ is defined as the sum of the user utilities; that is, $U(\mathbfit{X})=\sum_{k=1}^K u_k(\mathbfit{x}^k)$. The function $U$ has the following two properties:
\begin{enumerate} \label{assm:twicediff}
  \item $U(\mathbfit{X})$  is twice differentiable over $\Omega_1\times\ldots\times\Omega_K$
	\item $U(\mathbfit{X})$ is bounded from above over $\Omega_1\times\ldots\times\Omega_K$
\end{enumerate}
\end{assm}
\begin{rem}\label{remSingular} Many typical system utility functions (e.g., weighted proportional fairness, weighted harmonic mean) admit equivalent sum utility formulations that satisfy Assumptions \ref{ukProp} and \ref{Uprop}. An example is given in Appendix \ref{ExamplePF}. If the corresponding utility functions are not differentiable at $x_{k,k}=0$, then we have to restrict the optimization domain $\Omega_k$ as follows: For some $\mu\in\mathbb{R}_{++}$, we define
\begin{align}
\mathbb{R}_{k,\mu}^K=\left\{\left[y_1,\ldots,y_K\right]^T: y_k\in[\slfrac{1}{\mu},\infty],y_l\in\mathbb{R}_+,\forall l\neq k\right\}\nonumber
\end{align} 
and  $\Omega_{k,\mu}=\Omega_k \cap \mathbb{R}_{k,\mu}^K$. Consequently, the sum utility function $U$ is twice differentiable on the compact convex set $\Omega_{1,\mu}\times\ldots\times\Omega_{K,\mu}$. Note that $\lim_{\mu\rightarrow\infty} \Omega_{k,\mu} = \Omega_{k}$ so that this restriction becomes negligible for large $\mu$.
\end{rem}
\vspace{0.25cm}
The beamforming optimization problem is given by
\begin{align}
\max_{\mathbfit{X}} U(\mathbfit{X})\;\mathrm{s.\,t.}\;  \mathbfit{X} \in \Omega_1\times\ldots\times\Omega_K.\tag{$\mathrm{P}_0$} \label{bfOptProb} \nonumber 
\end{align}
As already mentioned, Problem \eqref{bfOptProb} is non-convex and NP-hard. Due to the convexity of $\Omega_k,\forall k$, we can formulate a necessary condition for the optimal solution to Problem \eqref{bfOptProb}, which is also a sufficient condition when the sum utility $U$ is convex with respect to $\mathbfit{X}$. 
\begin{prop}[Optimality Condition, \cite{bertsekas1995nonlinear} Proposition 2.1.2]\label{propStatCond}
If $\mathbfit{X}^*=[\mathbfit{x}_1^*,\ldots,\mathbfit{x}_K^*]$ is a local maximum of $U$ over $\Omega_1\times\ldots\times\Omega_K$, then we have
\begin{align}
\nabla_k U(\mathbfit{X}^*)\left(\mathbfit{x}-\mathbfit{x}_k^*\right)\leq 0,\forall \mathbfit{x}\in \Omega_k,\forall k, \label{eq:optcond1} \tag{$\mathrm{C}_0$}  \nonumber 
\end{align}
where $\nabla_{k} U$ denotes the gradient vector of $U$ with respect to $\mathbfit{x}_k$, given by
\begin{align}
\nabla_{k} U\left(\mathbfit{X}\right)=[\slfrac{\partial{U(\mathbfit{X})}}{\partial{x_{k,1}}}, \ldots,\slfrac{\partial{U(\mathbfit{X})}}{\partial{x_{k,K}}}]. \nonumber
\end{align}
\end{prop}
Note that Condition \eqref{eq:optcond1} is sufficient for all stationary points of Problem \eqref{bfOptProb}, which are of special interest because these can be easily found with a gradient-based algorithm. The next theorem establishes an important property of stationary points.
\begin{thm}[Sufficiency of Single-Stream Beamforming for Stationary Points]\label{thm2}
All stationary points $\mathbfit{X}^*=[\mathbfit{x}_1^*,\ldots,\mathbfit{x}_K^*]$ of Problem \eqref{bfOptProb} can be achieved using single-stream beamforming.
A set of corresponding beamforming vectors can be found as follows: Let $\left(\mathbfit{Q}_1^*,\ldots,\mathbfit{Q}_K^*\right)$ be the tuple of (possibly high rank) correlation matrices\footnote{These correlation matrices can be obtained by the scaled gradient projection algorithm as described in Section \ref{ADBF}.} that achieve the stationary point $\mathbfit{X}^*$. For each $k$, a corresponding beamforming vector $\mathbfit{w}_k^*$ can be approached as follows:
\begin{enumerate}
	\item If $\mathrm{rank}(\mathbfit{Q}_k^*)\leq 1$ then $\mathbfit{w}_k^*=\sqrt{\lambda_1(\mathbfit{Q}_k^*)}\cdot v_\mathrm{max}(\mathbfit{Q}_k^*)$.
	\item If $\mathrm{rank}(\mathbfit{Q}_k^*)>1$ then $\mathbfit{w}_k^*$ is given by the solution of the convex optimization problem
\begin{align}
\min_{\mathbfit{w}_k} &-\Re(\mathbfit{h}_{kk}^H\mathbfit{w}_k) \tag{$\mathrm{P}_1$} \label{findBfVecProb}  \nonumber \\
\mathrm{s.\,t.}\;  &\left|\mathbfit{h}_{kl}^H\mathbfit{w}_k\right|^2\leq x_{k,l}(\mathbfit{Q}_k^*),\forall l\neq k\nonumber\\
&\left\|\mathbfit{w}_k\right\|^2\leq 1.\nonumber
\end{align}
\end{enumerate}
\end{thm}
\begin{proof} The proof is provided in Appendix \ref{prfUtilOpt}.
\end{proof}
\begin{rem}
In our numerical experiments (see Section \ref{SimRes}), the first case (i.e., $\mathrm{rank}(\mathbfit{Q}_k^*)\leq 1$) always occurred; that is, the second case is mainly for the sake of mathematical completeness. Furthermore, if $\mathrm{rank}(\mathbfit{Q}_k^*)>1$ then we do not necessarily find the exact beamforming vector $\mathbfit{w}_k^*$ which generates the power gain vector $\mathbfit{x}_k(\mathbfit{Q}_k^*)$, as described in the proof. However, for the resulting sum utility $U$ it always holds that $U\left(\mathbfit{X}\left(\mathbfit{Q}_1^*,\ldots,\mathbfit{Q}_K^*\right)\right)\leq U(\mathbfit{X}(\mathbfit{w}_1^*(\mathbfit{w}_1^*)^H,\ldots,\mathbfit{w}_K^*(\mathbfit{w}_K^*)^H))$. The strict inequality can occur when $\mathbfit{X}^*$ is not a local maximum so that it may be possible  
\begin{enumerate}
	\item to increase the useful signal power $x_{k,k}$ in Problem \eqref{findBfVecProb} without violating the interference constraints $ x_{k,l}(\mathbfit{Q}_k^*),\forall l\neq k$,
	\item to satisfy at least one interference constraint $x_{k,l}(\mathbfit{Q}_k^*)$ in Problem \eqref{findBfVecProb} with strict inequality.
\end{enumerate}
By Assumption \ref{ukProp}, each case will yield an increased $u_k$ for some $k$ and thus an increased sum utility $U$.
\end{rem}

\section{Distributed Asynchronous Optimization} \label{ADBF}
The structure of Problem \eqref{bfOptProb} admits a (spatially) distributed implementation whereby the transmitters solve local subproblems and exchange interim computation results via a backhaul network. None of the transmitters possess all relevant information, and there exist communication delays between the transmitters. Following \cite{Bertsekas:1989:PDC}, an algorithm is said to experience a substantial \textit{synchronization penalty} if the waiting time due to communication delays as well as due to specific computation sequences is a sizable fraction of the total time needed to solve the problem. In that case, an asynchronous implementation can often substantially reduce the synchronization penalty because there is no requirement for waiting at predetermined points. Another advantage is that a global synchronization mechanism is not necessary.\\
We start with the derivation of the synchronous distributed implementation, which serves as a reference solution. Thereafter, we introduce the asynchronous computation model and elaborate on the algorithm's convergence and rate of convergence.
For ease of notation, we omit the dependence of $\mathbfit{x}_k$ on $\mathbfit{Q}_k$. In order to distinguish variable values at different time instants, we introduce the iteration index $n$ as an argument (e.g., the value of $\mathbfit{x}_k$ at time instant $n$ is denoted by $\mathbfit{x}_k(n)$).


\subsection{Synchronous Scaled Gradient Projection Algorithm}\label{syncSGP}
Due to the separability of the constraint set $\Omega_1\times\ldots\times\Omega_K$, we can split Problem \eqref{bfOptProb} into $K$ coupled subproblems, which are iteratively solved by the individual transmitters. The $k$-th subproblem at iteration index $n$ solves for the improved transmit power gain vector $\mathbfit{x}_k(n+1)$, and is given by
\begin{align}
\mathbfit{x}_k(n+1)=\arg\max_{\mathbfit{x}_k} U(\mathbfit{X}(n))\;\mathrm{s.\,t.}\;  \mathbfit{x}_k \in \Omega_k.  \label{eq_parOptUpdate}
\end{align}
The convergence of the sequences $\left\{\mathbfit{x}_k(n)\right\},\forall k$, generated by the nonlinear equation \eqref{eq_parOptUpdate}, can not be guaranteed because $U$ is generally non-convex (i.e., \eqref{eq_parOptUpdate} can not be formulated as a (pseudo-) contraction iteration\footnote{An iterative algorithm of the form $x(n+1)=T(x(n)),n=0,1,\ldots,$ is called contraction iteration if the mapping $T:\mathcal{X}\rightarrow\mathcal{X}$ has the property $\left\|T(x)-T(y)\right\|\leq \alpha \left\|x-y\right\|,\forall x,y\in\mathcal{X}$ with $\alpha\in[0,1)$.
Contraction iterations are of particular interest because there exists general results on the existence and uniqueness of fixed points, see \cite[Section 3.1]{Bertsekas:1989:PDC}.}).
However, convergence to a limit point can be established for linearized algorithms where the variable update is a linear function of $\nabla U(\mathbfit{X})$. Thus, we adopt the scaled gradient projection (SGP) method from \cite[Subsection 3.3.3]{Bertsekas:1989:PDC}, where the update for the $k$-th subproblem is described by the equation
\begin{align}
\mathbfit{x}_k(n+1)&=\left[\mathbfit{x}_k(n)+\gamma_k \mathbfit{M}_k^{-1} \boldsymbol{\lambda}_k(n)\right]_{\mathbfit{M}_k}^{\Omega_k}, \label{eq:updateStep}
\end{align}
using the step size parameter $\gamma_k$, the update direction 
\begin{align}
\boldsymbol{\lambda}_k(n)=&\left[\nabla_k U\left(\mathbfit{X}(n)\right)\right]^T\nonumber\\
=&\left[\frac{\partial{u_1(\mathbfit{x}^1(n))}}{\partial{x_{k,1}}},\ldots,\frac{\partial{u_K(\mathbfit{x}^K(n))}}{\partial{x_{k,K}}}\right]^T \label{eq:currPartDer}, 
\end{align}
and the diagonal\footnote{In general, $\mathbfit{M}_k$ is assumed to be positive definite, at least on a proper subspace in which all update steps \eqref{eq:updateStep} take place. Typically, $\mathbfit{M}_k$ would be chosen to approximate the Hessian matrix $\nabla_k^2 U\left(\mathbfit{X}(n)\right)$ such as done in the projected Jacobi method where $\mathbfit{M}_k$ is a diagonal matrix, with its diagonal entries equal to the diagonal entries of $\nabla_k^2 U\left(\mathbfit{X}(n)\right)$. For the considered sum utility problem, the Hessian is always a diagonal matrix. Furthermore, the use of diagonal scaling matrices facilitates the proof of Theorem \ref{thm:tsit}.} 
scaling matrix $\mathbfit{M}_k=\mathrm{diag}\left(\beta_{k,1},\ldots,\beta_{k,K}\right)$ with $\beta_{kl}\in\mathbb{R}_{++},\forall l$.
We use the notation $\left[\mathbfit{x}\right]_{\mathbfit{M}_k}^{\Omega_k}$ to denote the scaled projection (with respect to Euclidean norm) of the vector $\mathbfit{x}\in\mathbb{R}^K$ onto the convex set $\Omega_k$, see Subsection \ref{pgProject}.

As illustrated in Fig. \ref{fig:interactions}, the subproblems are coupled by the power gains and the partial derivatives, which are iteratively calculated and exchanged between the transmitters. For instance, the $l$-th component of the vector $\boldsymbol{\lambda}_k(n)$ is computed by the $l$-th transmitter, which in turn requires the knowledge of all receive power gains $x_{j,l},\forall j$. Consequently, each pair of transmitters $(k,l)$ needs to exchange the two (real-valued) power gains $(x_{k,l},x_{l,k})$ and the two (real-valued) derivatives $(\partial u_k/\partial x_{l,k},\partial u_l/\partial x_{k,l})$, to accomplish the local update \eqref{eq:updateStep}.

\begin{figure}[t]
\centering
\includegraphics[width=0.5\textwidth]{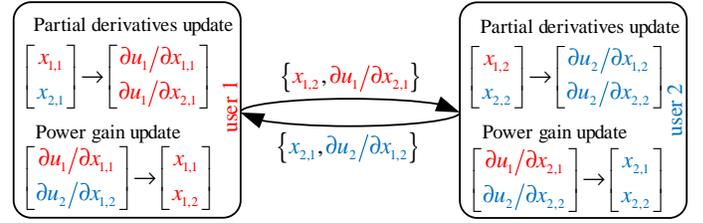}
\caption{Information exchange and dependencies for the two user case; i.e., $U=u_1(\mathbfit{x}^1)+u_2(\mathbfit{x}^2)$ with $\mathbfit{x}^1=[x_{1,1},x_{2,1}]$ and $\mathbfit{x}^2=[x_{1,2},x_{2,2}]$.}
\label{fig:interactions}
\end{figure}

The \textit{synchronous SGP} algorithm is summarized as follows:
\begin{enumerate}
	\item \textit{Initialization}: Each transmitter $k$ chooses an initial power gain vector $\mathbfit{x}_k(0) \in \Omega_k$. Set iteration index to $n=0$.
	\item \textit{Gradient Update}: Each transmitter $k$ calculates the set of current partial derivatives ${\partial{u_k(\mathbfit{x}^k(n))}}/{\partial{x_{l,k}}},\forall l$, and communicates the $l$-th element to the $l$-th transmitter.
	\item \textit{Update Step}: Each transmitter $k$ calculates the new power gain vector $\mathbfit{x}_k(n+1)$ according to \eqref{eq:updateStep}, and communicates the $l$-th component of $\mathbfit{x}_k(n+1)$ to the $l$-th transmitter.
	\item Increment $n$ and repeat from step 2).
\end{enumerate}
If all transmitters wait until they have acquired the most recent information and perform their update steps concurrently at the same iteration index then the algorithm is mathematically equivalent to the centralized SGP method (cf. \cite[Equation (3.6)]{Bertsekas:1989:PDC}) and the corresponding convergence result is applicable:
\begin{prop}[Convergence of the SGP Algorithm, \cite{Bertsekas:1989:PDC} Proposition 3.7 (h)]
If the step size $\gamma$ is chosen small enough, then any limit point $\mathbfit{X}^*=[\mathbfit{x}_1,\ldots,\mathbfit{x}_K^*]$ of the sequence $\left\{\mathbfit{X}(n)\right\}$, generated by the centralized SGP algorithm, satisfies the stationarity conditions \eqref{eq:optcond1}. If $U$ is also convex on the set $\Omega_1\times\ldots\times\Omega_K$ then $\mathbfit{X}^*$ is the global maximizer. 
\end{prop}
A proper condition for the step size parameters $\gamma_k,\forall k$ is formulated in the next subsection by Theorem \ref{thm:tsit}.

Note that the computation of the projection $\mathbfit{x}_k(n+1)=[\cdot]_{\mathbfit{M}_k}^{\Omega_k}$ is accomplished over the convex set $\mathcal{Q}$ of correlation matrices (see Subsection \ref{pgProject}), and produces a correlation matrix $\mathbfit{Q}_k(n+1)$. By Theorem \ref{thm2}, we can extract a corresponding beamforming vector $\mathbfit{w}_k(n+1)$, which can be applied for data transmission while the optimization process is still in progress.

\begin{rem}
Similar to the proposed SGP algorithm, the DP algorithm \cite{4797605} utilizes the partial derivatives of $U$ with respect to the power gains. There, the derivatives are negated and called interference prices. 
\end{rem}

\subsection{Asynchronous SGP Algorithm}\label{asyncSGP}
We now adopt the asynchronous computation model from \cite{1104412}, in which each transmitter does \emph{not} need to communicate to each other transmitter at each time instant; also the transmitters may perform their updates at different intervals and they may keep performing without having to wait until they receive messages that have been transmitted to them. Thus, they perform their updates with possibly outdated information.\\
For analysis purposes only, we consider a global event-driven clock that indexes all events of interest (such as an update step, transmission or reception of a message) by a discrete variable $n\in \mathbb{N}_0$, which is called the time index. 
Furthermore, we define sets of time indices at which each user updates its power gain vector or partial derivatives. These sets need not be known to any of the users; that is, their knowledge is not required to compute an update.
\begin{defn}[Set of Update Times]\label{PXdefnUpdTimes}
Let $\mathcal{X}_k$ (resp. $\mathcal{Y}_k$) be the unbounded set of time indices at which $\mathbfit{x}_k$ (resp. $\partial u_k/\partial x_{l,k},\forall l$) is updated by the $k$-th user.
\end{defn}
\begin{rem}
The convergence analysis of the asynchronous SGP algorithm relies on the un-boundedness of the sets $\mathcal{X}_k$ and $\mathcal{Y}_k$; that is, theoretically the algorithm never stops updating its variables. In practice, a stopping criterion is required (e.g., on the number of iterations). Moreover, an alternative to the unbounded set assumption is to bound the inter-update interval by a constant, as proposed in \cite[Chapter 7]{Bertsekas:1989:PDC}.
\end{rem}

There is no explicit notion of a processing period for the update computations. Without loss of generality, we index the time instant when an update computation starts by $n$, and assume that it is completed at index $n+1$ (i.e., there occurs no indexed event in between). Moreover, we assume that at index $n+1$ a message with the updated value is sent to the other users.
Any effective processing period can be accounted for in the difference between time index $n+1$ and the time index of a received message which contains the updated value.

\begin{defn}[Communication Delays]\label{defn1}
Let $q_{k,l}(n)\in \mathbb{N}_0,0\leq q_{k,l}(n)\leq n$ (resp. $p_{k,l}(n)\in \mathbb{N}_0,0\leq p_{k,l}(n)\leq n$) be the time index of a message with a value of $\partial{u_k}/\partial{x_{l,k}}$ (resp. $x_{k,l}$) that was sent from transmitter $k$ to transmitter $l$, and this was the last such message received not later than at time index $n$.
Without loss of generality, we assume $q_{k,k}(n)=p_{k,k}(n)=n,\forall k$; that is, each user modifies its (local) data exclusively, so that no communication delays arise. 
\end{defn}
Based on this definition, each user $k$ has the following local information at time index $n$:
\begin{align}
\boldsymbol\chi^k(n) = &\left[x_{1,k}\left(p_{1,k}(n)\right),\ldots,x_{K,k}\left(p_{K,k}(n)\right)\right], \label{eq:delayedPGvec}\\
\boldsymbol\lambda_k(n)=&\left[\frac{\partial{u_1\left(\boldsymbol{\chi}^1\left(q_{1,k}(n)\right)\right)}}{\partial{x_{k,1}}},\ldots,\frac{\partial{u_K\left(\boldsymbol{\chi}^K\left(q_{K,k}(n)\right)\right)}}{\partial{x_{k,K}}}\right]^T.\label{eq:delayedPartDer}
\end{align}

The \textit{asynchronous SGP} algorithm is summarized as follows: 
\begin{enumerate}
	\item \textit{Initialization}: Each transmitter $k$ chooses an initial power gain vector $\mathbfit{x}_k(0)\in \Omega_k$. Set time index $n=0$ and subsequently increment $n$ in arbitrary intervals.
	\item  \textit{Update Steps}:	
	\begin{itemize}
		\item If $n\in\mathcal{Y}_k$ then the $k$-th transmitter calculates the set of the partial derivatives ${\partial{u_k}}/{\partial{x_{l,k}}},\forall l$, using the received power gain vector defined in \eqref{eq:delayedPGvec}, and sends (at time index $n+1$) the $l$-th element to the $l$-th transmitter.
		\item If $n\in\mathcal{X}_k$ then the $k$-th transmitter calculates $\mathbfit{x}_k(n+1)$ according to \eqref{eq:updateStep} and \eqref{eq:delayedPartDer},	and communicates the $l$-th component of $\mathbfit{x}_k(n+1)$ to the $l$-th transmitter. 
	\end{itemize}	
\end{enumerate}

As shown in \cite[Subsection 6.3.2, Example 3.1]{Bertsekas:1989:PDC}, gradient algorithms require finite delays to ensure the convergence. Such algorithms are called \textit{partially asynchronous}. 
\begin{assm}[Finite Communication Delays] \label{assm:1}
For some (finite) constants $Q_{l,k}$ and $P_{l,k}$, we have for all $l,k,n$
\begin{align}
n-Q_{l,k}\leq q_{l,k}(n)\leq n,\nonumber\\
n-P_{l,k}\leq p_{l,k}(n)\leq n.\nonumber
\end{align}
\end{assm}
The delay bounds require a-priori knowledge of the network, and can be determined e.g., during system design.

Since the partial derivatives are subject to delays, we have to quantify their rate of change with respect to $\mathbfit{X}$, which yields the following assumption with respect to the curvature of $U$.
\begin{assm}[Curvature Bounds] \label{assm:2}
For the second-order partial derivatives, there exist bounds\footnote{By definition of $U$ we have $\slfrac{\partial^2 U}{\partial x_{k,l}\partial x_{s,t}}\left(\mathbfit{X}\right)=0$ for $l\neq t$. So, there is no need for a fourth subindex $t$ such as $K_{kl,st}$.} $K_{kl,s}$ such that 
\begin{align}
\left|\frac{\partial^2 U\left(\mathbfit{X}\right)}{\partial x_{k,l}\partial x_{s,l}}\right|\leq K_{kl,s},\forall \mathbfit{x}_k\in\Omega_k,\forall k,l,s\nonumber
\end{align}
\end{assm}
By Assumption \ref{Uprop} these bounds always exist. Note that instead of using the Lipschitz constants of $\nabla_k U,\forall k$ as an upper bound (as proposed in \cite[Subsection 7.5.1]{Bertsekas:1989:PDC}), the bounds $K_{kl,s}$ can be explicitly determined with moderate effort and allow a more detailed description of the interactions between the subproblems. An example is given in Appendix \ref{ExamplePF}.\\
Next, we give sufficient conditions for the step size parameters $\gamma_k$ which ensure, that the asynchronous SGP algorithm converges to a stationary point of Problem \eqref{bfOptProb}. 
\begin{thm}[Step Size Bounds] \label{thm:tsit}
Suppose that for each transmitter $k$ we have
\begin{align}
\gamma_k < \mathrm{min}_l\, \slfrac{2\beta_{k,l}}{D_{k,l}}\label{eq:suffcond}
\end{align}
with $D_{k,l}=\sum_{s=1}^K K_{kl,s}(1+P_{s,l}+Q_{l,k})+K_{sl,k}\left(P_{k,l}+Q_{l,s}\right)$,
then any limit point $\mathbfit{X}^*$ of the sequence $\left\{\mathbfit{X}(n)\right\}$, generated by the asynchronous SGP algorithm, satisfies the stationary conditions $(\mathrm{C}_0)$.
\end{thm}
\begin{proof} The proof is provided in Appendix \ref{proof:tsit}.
\end{proof}
As remarked in \cite[Subsection 5.6]{2140467}, the step size bounds are sufficient for convergence but they are not tight, nor necessary. Since the convergence rate is governed by the smallest and largest eigenvalues of the transformed Hessian $\mathbfit{M}_k^{-\frac{1}{2}}\nabla_k^2 U \mathbfit{M}_k^{-\frac{1}{2}}$ (\cite[Section 2.3.1]{bertsekas1995nonlinear}), one should try to choose the scaling matrix $\mathbfit{M}_k$ as close as possible to the Hessian matrix $\nabla_k^2 U$. This is achieved by setting $\beta_{k,l} = D_{k,l}$, for which we obtain a common upper bound on the step sizes; that is, $\gamma_k < 2,\forall k$. Then, each element $\beta_{k,l}$ of the scaling matrix $\mathbfit{M}_k$ acts as a step size parameter for the $l$-th component of the gradient vector $\boldsymbol \lambda_k$.
Further mechanisms for improving the convergence rate are discussed in the Subsection \ref{convSpeedUp}.

\subsection{Scaled Projection onto the Power Gain Region}\label{pgProject}
Next, we show how to accomplish the projection step in \eqref{eq:updateStep}.
By \cite[Proposition 3.7 (a)]{Bertsekas:1989:PDC} the projection $\mathbfit{x}_k(n+1)=\left[\mathbfit{z}\right]_{\mathbfit{M}_k}^{\Omega_k}$ is unique and given by
\begin{align}
\left[\mathbfit{z}\right]_{\mathbfit{M}_k}^{\Omega_k}=&\arg\min \left\|\mathbfit{z}-\mathbfit{x}\right\|_{\mathbfit{M}_k}^2  \tag{$\mathrm{P}_2$}\label{projProb}\nonumber\\
&\mathrm{s.\,t.}\; \mathbfit{x}\in \Omega_k. \nonumber
\end{align}
By rewriting $x_{k,l}\left(\mathbfit{Q}\right) =\mathrm{tr}(\mathbfit{Q}\mathbfit{h}_{kl}\mathbfit{h}_{kl}^H)$, the weighted inner product can be formulated as
\begin{align}
\left\|\mathbfit{z}-\mathbfit{x}\right\|_{\mathbfit{M}_k}^2&=\left(\mathbfit{z}-\mathbfit{x}_k\left(\mathbfit{Q}\right)\right)^T \mathbfit{M}_k \left(\mathbfit{z}-\mathbfit{x}_k\left(\mathbfit{Q}\right)\right)\nonumber\\
&=\frac{1}{2}\mathrm{tr}\left(\varphi(\mathbfit{Q})\mathbfit{Q}\right) +  
\mathrm{tr}\left( \mathbfit{C}\mathbfit{Q}\right) +\left\|\mathbfit{z}\right\|_{\mathbfit{M}_k}^2 \label{eq:equiProjObjFun}
\end{align}
with the self-adjoint positive semi-definite linear operator $\varphi(\mathbfit{Q})=2\sum_{l=1}^K \beta_{k,l}\mathbfit{H}_{kl} \mathbfit{Q} \mathbfit{H}_{kl}$, and matrices $\mathbfit{C}=-2\sum_{l=1}^K \beta_{k,l}z_l\mathbfit{H}_{kl}$ and $\mathbfit{H}_{kl}=\mathbfit{h}_{kl}\mathbfit{h}_{kl}^H$. Thus, the minimization over the power gain region $\Omega_k$ can be equivalently accomplished over the convex set $\mathcal{Q}$. The resulting (convex) quadratic semi-definite program (QSDP) is given by
\begin{align}
&\min \frac{1}{2}\mathrm{tr}\left(\varphi(\mathbfit{Q})\mathbfit{Q}\right) +  
\mathrm{tr}\left( \mathbfit{C}\mathbfit{Q}\right) \tag{$\mathrm{P}_3$}\label{qsdpProb}\nonumber\\
&\mathrm{s.\,t.}\; \mathrm{tr}\left(\mathbfit{Q}\right) \leq 1,\;\mathbfit{Q} \succeq 0. \nonumber
\end{align}
The global optimal solution $\mathbfit{Q}^*$ can be found efficiently by a QSDP solver, e.g. the \textsc{Matlab} software QSDP-0 \cite{QSDP0}. The solution of Problem \eqref{projProb} is obtained by 
\begin{align}
\mathbfit{x}_k(n+1)=\left[\mathbfit{z}\right]_{\mathbfit{M}_k}^{\Omega_k}=\mathbfit{x}_k(\mathbfit{Q}^*).
\end{align}

\begin{rem}[Solution by Gradient Projection Method]
Alternatively, Problem \eqref{projProb} can be solved iteratively with the gradient projection method \cite[Section 2.3]{bertsekas1995nonlinear}. Therefore, we minimize \eqref{eq:equiProjObjFun} over the convex cone of positive semi-definite matrices, subject to the linear inequality constraint $\mathrm{tr}(\mathbfit{Q})\leq 1$. The projection onto the constraint set is accomplished by an appropriate scaling of the eigenvalues of $\mathbfit{Q}$ (see \cite[Section 8.1.1]{Boyd}, \cite[Section IV-C]{1237413}).
\end{rem}

\subsection{Improving the Convergence Rate} \label{convSpeedUp}
In the following, we describe two mechanisms that improve the convergence rate of the asynchronous SGP algorithm, and which preserve its convergence to a stationary point of \eqref{bfOptProb}.

\subsubsection{Speed-up ${S_1}$ (Normalized Power Gains)}
The first speed-up mechanism tightens the bounds for the step size parameters $\gamma_k,\forall k$, by exploiting the fine structure of the problem (i.e., the channel coupling strength between the users).
Therefore, the SGP algorithm is formulated in the linearly transformed optimization domain $\Omega_1'\times\ldots\times\Omega_K'$ with $\Omega_k'=\{\mathbfit{x}_k'=\mathbfit{T}_k^{-1}\mathbfit{x}_k:\mathbfit{x}_k\in \Omega_k\},\forall k$ and $\mathbfit{T}_k=\mathrm{diag}(\left\|\mathbfit{h}_{k1}\right\|^2,\ldots,\left\|\mathbfit{h}_{kK}\right\|^2)$.
Consequently, each curvature bound $K_{kl,s}$ is scaled by the corresponding channel gains, yielding $K_{kl,s}'=\left\|\mathbfit{h}_{kl}\right\|^2 \left\|\mathbfit{h}_{sl}\right\|^2 K_{kl,s},\forall k,l,s$. One should note that small channel gains scale down the curvature bounds and thus yield a tighter lower bound in \eqref{PXeq:taylorLB} for the quadratic term of the second-order Taylor expansion of $U$. The resulting convergence speed-up is illustrated in Section \ref{SimRes}.

\subsubsection{Speed-up ${S_2}$ (Adaptive Curvature Bounds)}
The basic idea of the second speed-up mechanism is to adapt the curvature bounds during the optimization process. One should note that the global curvature bounds $K_{lk,s},\forall, l,k,s$, as formulated in Assumption \ref{assm:2}, reflect the \emph{worst case} curvature of the sum utility function $U$. For the majority of operating points, these bounds are too stringent and cause a slow convergence speed. The proposed speed-up mechanism relies on the following assumption, which is satisfied by e.g., the sum rate and proportional fair rate utility.

\begin{assm}[Monotonicity of the Curvature Bounds] \label{P2monCurvAssm} 
For all $k,l,s$, the absolute value of the second-order partial derivative $\slfrac{\partial^2 U}{\partial x_{l,k}\partial x_{s,k}}$ is a monotonic function with respect to the power gains $x_{l,k},\forall l$.
\end{assm}

Let $\mathcal{Z}_k$ be the unbounded set of time indices when the $k$-th transmitter updates its curvature bounds $K_{lk,s},\forall l,s$. A suitable choice for the set $\mathcal{Z}_k$ is given by the set of time indices, at which user $k$ receives power gain messages from the other users. By doing so, every change in the operating point is tracked immediately (but subject to communication delays).

The asynchronous SGP algorithm is extended as follows:
\begin{enumerate}
	\item \textit{Initialization}: Each transmitter $k$ maintains an upper bound $\hat{x}_{l,k}$ and a lower bound $\check{x}_{l,k}$ for every received power gain $x_{l,k},\forall l$. These bounds are initialized with the smallest and largest feasible value; that is, the $\hat{x}_{l,k}(0)=0$ (ZF beamforming) and $\check{x}_{l,k}(0)=\left\|\mathbfit{h}_{lk}\right\|^2$ (MRT beamforming). Based on these bounds and the monotonicity properties\footnote{A monotonic function attains its maximum at the boundary of its domain.} of the second-order partial derivatives (i.e., increasing or decreasing), transmitter $k$ calculates the initial upper bounds $K_{lk,s}(0),\forall l,s$, which are then communicated to the corresponding transmitters $l$ and $s$.
	\item \textit{Update Steps}: 	
	\begin{itemize}
		\item If $n\in\mathcal{Z}_k$ then transmitter $k$ updates the upper and lower bounds for the power gains according to $\hat{x}_{l,k}(n+1)=\max(\hat{x}_{l,k}(n),x_{l,k}(n)),\forall l$ and $\check{x}_{l,k}(n+1)=\min(\check{x}_{l,k}(n),x_{l,k}(n)),\forall l$. Based on these updated bounds, it adapts the curvature bounds $K_{lk,s}(n+1),\forall l,s$, which are then communicated to transmitters $l$ and $s$.
		\item The power gain update step at the $k$-th transmitter follows \eqref{eq:updateStep} but with the adapted step size parameter $\gamma_k(n)$ and scaling matrix $\mathbfit{M}_k(n)$. Both have been updated based on the (received) curvature bounds $K_{kl,s}(q_{l,k}(n)),\forall l,s$.
	\end{itemize}	
\end{enumerate}

Basically, the extended SGP algorithm can be understood as a second-order algorithm; that is, an algorithm which utilizes a second-order Taylor approximation for every operating point. The next proposition establishes the convergence of the extended SGP algorithm.

\begin{prop}[Asymptotic Convergence of the SGP Algorithm with Speed-Up $S_2$]
If the sum utility function $U$ satisfies Assumption \ref{P2monCurvAssm}, then the asynchronous SGP algorithm with adaptive curvature bounds converges asymptotically to a stationary point of Problem \eqref{bfOptProb}. 
\end{prop}
\begin{proof}
For each $l,k,s$, the curvature bound sequence $\{K_{lk,s}(n)\}$ converges to a limit point
$\lim_{n\rightarrow \infty}K_{lk,s}(n) = K_{lk,s}^*$, because it is upper bounded (cf. Assumption \ref{assm:2}) and never decreased by an update step. At the joint limit point of all sequences $\{K_{lk,s}(n)\},\forall l,k,s$, Assumption \ref{P2monCurvAssm} ensures\footnote{Assumption \ref{P2monCurvAssm} ensures that there are no local maxima of the curvature function, which can be overlooked by the algorithm.} that the curvature bounds are valid for the sequence $\{\mathbfit{X}(n)\}$, generated by the SGP algorithm. (If not, the joint limit point has not been reached.) Thus, the sequence $\{\mathbfit{X}(n)\}$ converges to a limit point which satisfies the stationarity condition \eqref{eq:optcond1} (Theorem \eqref{thm:tsit}).
\end{proof}

\begin{rem}[Non-Monotonic Convergence]
Although the convergence to a stationary point is guaranteed, we do not have \emph{monotonic} convergence in terms of the $U$. Every change (i.e., increase) in the curvature bounds implies a preceding power gain update step that has been based on incorrect curvature bounds, and which possibly decreased the sum utility.
\end{rem}

\section{Numerical Simulations}\label{SimRes}
To demonstrate the relative performance of different beamforming algorithms, we present numerical simulations for a small MISO IFC with different backhaul network topologies. Our interest lies on the overall processing time needed for convergence by an algorithm, which primarily depends on the number of iterations, update cycles, communication delays and synchronization periods. The following subsection provides a detailed description of these factors.

\subsection{Simulation Model}
\subsubsection{Simulation Clock}
For simulation purposes, we employ a dimensionless discrete clock $n\in\mathbb{N}_0$ with clock period $1$. This clock is used as a timeline for our simulations, and to quantify update cycles and communication delays.
Our special focus is on the communication bottleneck of different backhaul networks. Therefore, we assume that the duration of the update computations is negligible (i.e., zero with respect to our simulation clock), which is in contrast to the event-driven clock in Section \ref{asyncSGP}. This assumption implies a sufficiently large processing power at the transmitters.

\subsubsection{Channel Model}
We consider a MISO IFC with $K=4$ users, where each transmitter has $N=2$ antennas. The variance of the additive white Gaussian noise is $\sigma^2=10^{-2}$, yielding 20dB SNR at transmitter side. 
The elements of each channel vector $\mathbfit{h}_{kl}$ are independent circularly symmetric complex Gaussian random variables with zero mean and variance $\sigma_{k,l}^2$. The variance is used to model a 1-dimensional local coupling of neighboring users, and depends on the user indices as follows:
\begin{align}
\sigma_{k,l}^2=\left\{\begin{array}{ll}
1 & ,\,\mathrm{if}\;k=l, \\
10^{1-2\left|k-l\right|} & ,\,\mathrm{if}\;k\neq l.
\end{array} \right. \label{eq:chanCoupl}
\end{align}
One should note that the edge users $0$ and $3$ have only one interferer above noise level, while the inner users $1$ and $2$ have two interferers above noise level.

\subsubsection{Network Objective}
We assume the \emph{achievable rate} utility for each user $k$, given by $u_k(\mathbfit{x}^k) = \log_2\left(1 + \Gamma_k(\mathbfit{x}^k)\right)$ with $\Gamma_k(\mathbfit{x}^k)=x_{k,k}/(\sum_{l\neq k} x_{l,k} + \sigma^2)$. The system utility function is the proportional fair rate, which is defined as
\begin{align}
U^\mathrm{pf}(\mathbfit{X})=\prod_k u_k^{\frac{1}{K}}(\mathbfit{x}^k).\label{eq:PFutil}
\end{align}
As shown in Appendix \ref{ExamplePF}, this system utility function admits an equivalent sum utility problem formulation.

\begin{figure}[t]
\centering
\includegraphics[width=0.48\textwidth]{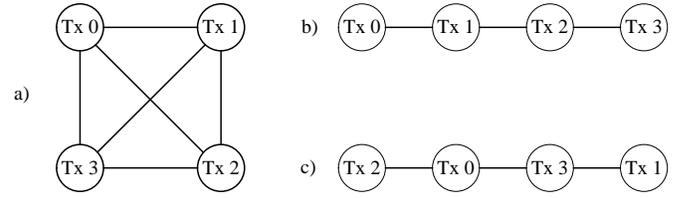}
\caption{Illustration of a) mesh backhaul network, b) daisy chain backhaul network and c) permuted daisy chain backhaul network.}
\label{PXfig:backhaulNW}
\end{figure}

\begin{table*}[!t]
\small
\begin{tabular}{|l|l|l|}
     \hline
			             & NW1  &  NW2   \\
      \hline
			\hline
			async. SGP   &  \multicolumn{2}{c|}{ $\mathcal{X}_k=\mathcal{Y}_k=\{Tn: n \in\mathbb{N}_0\}$}  \\
			\cline{2-3}
			with $T\in\mathbb{N}$ & $n-p_{k,l}(n)=n-q_{k,l}(n)=\left\lfloor \frac{D}{T}\right\rfloor$ & $n-p_{k,l}(n)=n-q_{k,l}(n)=\left\lfloor \frac{D\left|k-l\right|}{T}\right\rfloor$ \\
			\hline	
			sync. SGP  & $\mathcal{X}_k=\{Tn: n \in\mathbb{N}_0\}$, $\mathcal{Y}_k=\{Tn+D: n \in\mathbb{N}_0\}$& $\mathcal{X}_k=\{Tn: n \in\mathbb{N}_0\}$, $\mathcal{Y}_k=\{Tn+(K-1)D: n \in\mathbb{N}_0\}$\\
									& with $T=2D$ & with $T=2D(K-1)$\\
									\cline{2-3}
									 & \multicolumn{2}{c|}{$n-p_{k,l}(n)=n-q_{k,l}(n)=0$}   \\
			\hline	
			sync. DP & $\mathcal{X}_k=\{Tn+\varphi_k: n \in\mathbb{N}_0\}$, $\mathcal{Y}_k=\bigcup_l \mathcal{X}_l$, & $\mathcal{X}_k=\{T n+\varphi_k: n\in\mathbb{N}_0\}$, $\mathcal{Y}_k=\bigcup_l \mathcal{X}_l$\\ 
			with $M\in\mathbb{N}_0$	 & with $T=(D+M)K$ and $\varphi_k=(D+M)k$ & with $T=\left[2\sum_{l=1}^{\left\lfloor \slfrac{K}{2}\right\rfloor}(K-l)+\left\lfloor \frac{K}{2}\right\rfloor\right]D+KM\cdot(K \bmod 2)$,\\								
										&& $\varphi_k=kM+D\left[\sum_{l=0}^{k-1}(K-l-1)+\max(0,2k-K-1)\right]$\\
										\cline{2-3}
		                 & \multicolumn{2}{c|}{$n-p_{k,l}(n)=n-q_{k,l}(n)=0$}   \\
			\hline	
\end{tabular}
\caption{Algorithm configurations for user $k\in\{0,\ldots,K-1\}$, with backhaul delay $D$, update period $T$, 'measurement and reporting' period $M$. The configurations for NW3 are obtained by permuting the user indices in the NW2 case.\label{PX:tab:AlgConfig}}
\end{table*}

\subsubsection{Backhaul Network}
The transmitters are able to exchange information via the backhaul network. Depending on the network topology, diverse communication delays arise. As shown in Figure \ref{PXfig:backhaulNW}, we evaluate three network topologies:
\begin{enumerate}[a)]
	\item NW1 - Mesh Network: Every pair of transmitters has a direct communication link. We assume that each link introduces a communication delay of $D$ clock cycles.
	\item NW2 - Linear Daisy Chain Network: Only neighboring transmitters (i.e., transmitters whose indices differ by one) possess a backhaul link with a delay of $D$ clock cycles. We assume that messages are forwarded so that every pair of transmitters is able to exchange information. The overall delay between transmitters $k$ and $l$ is given by $\left|k-l\right|D$. 
	\item NW3 - Permuted Daisy Chain Network: This network structure serves as a reference case, and is derived from BHNW2 by permuting the transmitter indices but keeping the channel coupling \eqref{eq:chanCoupl}.
\end{enumerate}

\subsubsection{Algorithm Configuration}
Next, we describe the evaluated algorithms in terms of their possible parametrizations and the user's timing behavior, which is assumed to be equal for all users. Following Section \ref{asyncSGP}, the timing behavior of the $k$-th user is characterized by the sets of update times $\mathcal{X}_k$, $\mathcal{Y}_k$ and the \emph{effective} communication delays $n-p_{k,l}(n)$ and $n-q_{k,l}(n),\forall n,l$. Here, the elements of $\mathcal{X}_k$, $\mathcal{Y}_k$ represent time instants of our simulation clock. We assume that every variable update is sent immediately to the other users. Thus, the effective communication delays reflect the number of update computations by which a received message is outdated.\\
Table \ref{PX:tab:AlgConfig} summarizes the sets of update times and effective communication delays for the evaluated algorithms. Note that:

\begin{itemize}
	\item For the \emph{synchronous SGP} algorithm, the update equation \eqref{eq:updateStep} requires current gradient information. Therefore, every power gain update is sent to the other transmitters, which in turn feed back their updated partial derivatives. We assume that the derivatives are updated immediately after the reception of a power gain message. Thus, the shortest feasible power gain update cycle $T$ is determined by the largest round trip delay between the transmitters. 
 \item The \emph{asynchronous SGP} algorithm can cope with outdated gradient information. Thus, its power gain update cycle $T$ can be chosen arbitrarily. Here, we assume that all power gains and partial derivatives are updated concurrently.
\item As a reference case, we simulate the synchronous \emph{distributed pricing} (DP) algorithm \cite{5205801}, which requires sequential power gain updates based on current gradient information. The exchange of the power gains is accomplished via radio transmission. Based on the received signals, receiver $k$ calculates its partial derivative and reports the value to transmitter $k$. Then, every transmitter sends its value via backhaul to the certain transmitter that will perform the next update step. We assume a round-robin update sequence for the transmitters. The signaling of the power gains is not subject to any ascertainable communication delays, but the 'measurement and reporting' task requires $M$ clock cycles.
\item As a second reference case, we simulate the (synchronous) \emph{cyclic coordinate descent} (CCD) algorithm \cite{5638157}, which has the same timing behavior as the DP algorithm. 
\end{itemize}
At time instant $n=0$, the algorithms are initialized with power gain vectors that correspond to the MRT beamformers. Furthermore, we assume that every transmitter has \emph{current} knowledge of the system state (i.e., partial derivatives, curvature bounds, etc.), and performs its first power gain update.\\
Moreover, we use the constant step size parameters $\gamma_k=1.99,\forall k$ for the SGP algorithm. The scaling matrices $\mathbfit{M}_k,\forall k$ are chosen as described in Section \ref{asyncSGP}, and require the knowledge about the curvature bounds $K_{kl,s},\forall k,l,s$. In Appendix \ref{ExamplePF}, we illustrate their computation for the proportional fair rate utility. Due to the singularity of the transformed utility functions $\log(u_k)$ at $x_{k,k}=0$, we employ the restricted optimization domains $\Omega_{k,\mu},\forall k$ (see Remark \ref{remSingular}) with\footnote{A smaller parameter $\mu$ would yield very conservative curvature bounds, which significantly slow down the SGP algorithm. However, very small direct link gains $x_{k,k}<0.1$ are unlikely to occur because such operating points are repulsive for the proportional fair rate utility.} $\mu=0.1$. For the SGP algorithm speed-up $S_2$, we assume that the required sets of update times are given by $\mathcal{Z}_k=\mathcal{Y}_k,\forall k$.
The projection problems within the SGP update steps are solved with the QSDP-0 solver \cite{QSDP0} and the accuracy tolerance $10^{-6}$.

\subsection{Simulation Results} 
\subsubsection{Comparison of different Speed-up Options}\label{sec:speedups}
We start with the synchronous SGP algorithm and illustrate the effect of the proposed speed-up methods $S_1$ and $S_2$ (see Subsection \ref{convSpeedUp}). We assume a mesh backhaul network with communication delay $D=1$. For this setup, the synchronous SGP algorithm calculates a power gain update every second clock cycle. As an upper performance bound, we compute the optimal utility value by using the branch-reduce-and-bound (BRB) algorithm from \cite{6129540} with the accuracy parameter $\epsilon=10^{-2}$.\\
Figure \ref{P2fig:convSpeedup} shows the system utility $U^\mathrm{pf}$ as a function of the time index for an exemplary channel realization. The 'plain' SGP algorithm converges very slowly due to the conservative curvature bounds and thus loose step size bounds\footnote{One should note that the actual step size parameters are fixed to $\gamma_k=1.99,\forall k$. However, for each user $k$, the elements of the diagonal scaling matrix $\mathbfit{M}_k$ play the role of step size parameters, one for each component of the gradient vector $\boldsymbol \lambda_k$.}. By using normalized power gains (speed-up $S_1$), which result in tighter step size bounds, the time needed for convergence can be reduced by approximately one order of magnitude. However, the algorithm is still rather slow. A significant convergence speed-up is achieved by combining the speed-up options $S_1$ and $S_2$. By adapting the curvature bounds to the most recent operating point, the SGP algorithm becomes a second-order algorithm which accomplishes reasonably large update steps. 
This comes at the cost of the additional calculation and exchange of the curvature bounds. Moreover, it should be noted that this speed-up may cause the convergence to a different limit point because the basin of attraction of a local maximum can potentially be left.
Due to its superior performance, we subsequently focus on the SGP algorithm with both speed-up options\footnote{The DP and CCD algorithm do not calculate and exchange any second-order derivatives. However, convergence of the DP algorithm relies on specific properties of the second-order derivatives. For the CCD algorithm, it is assumed that all utility functions are known to the users.} $S_1\&S_2$.

\begin{figure}[t]
\centering
\includegraphics[width=0.5\textwidth]{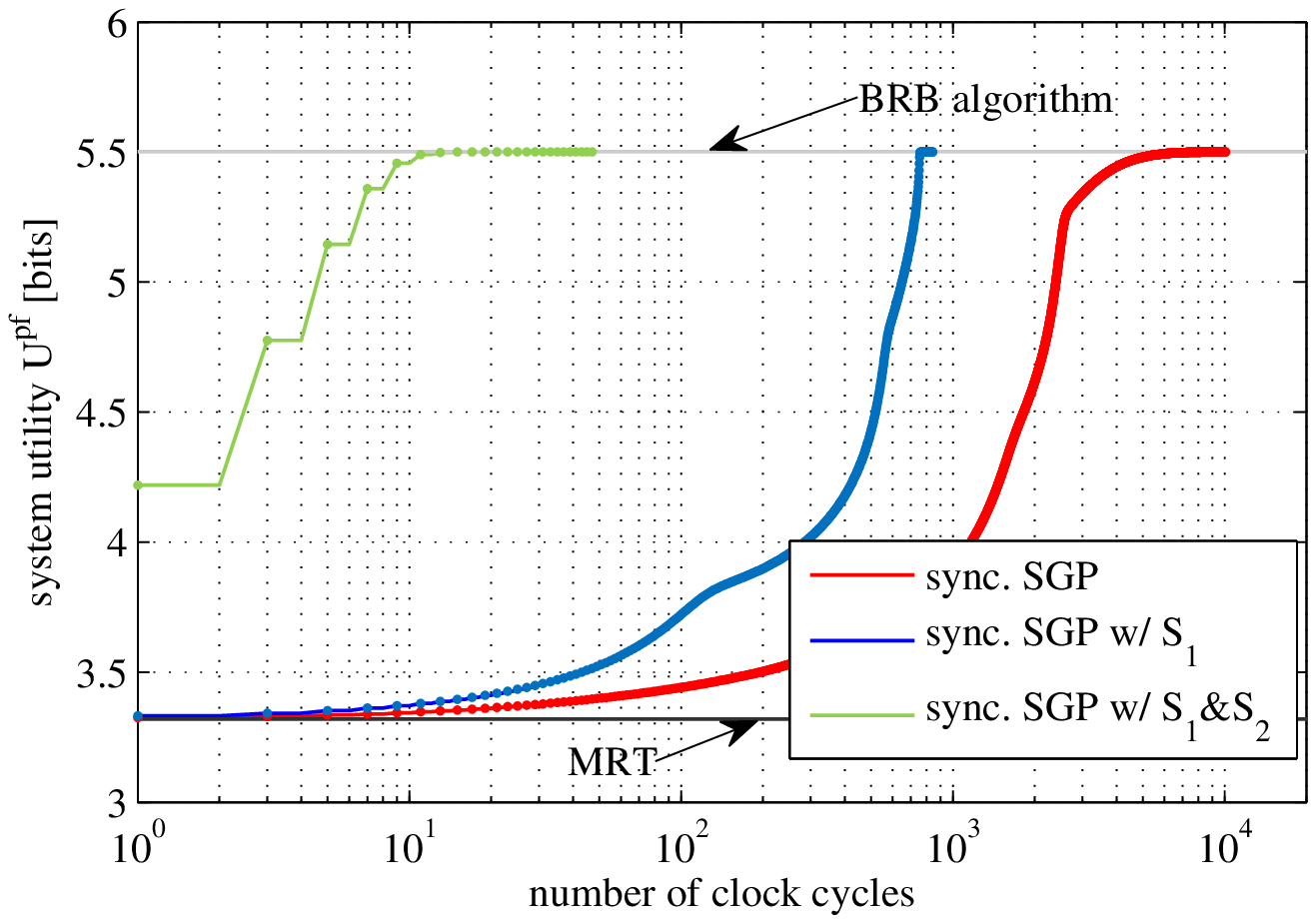}
\caption{Comparison of different speed-up methods for the sync. SGP. The plot shows the system utility $U^\mathrm{pf}$ as a function of the time index for a mesh backhaul network with communication delay $D=1$.}
\label{P2fig:convSpeedup}
%
\includegraphics[width=0.5\textwidth]{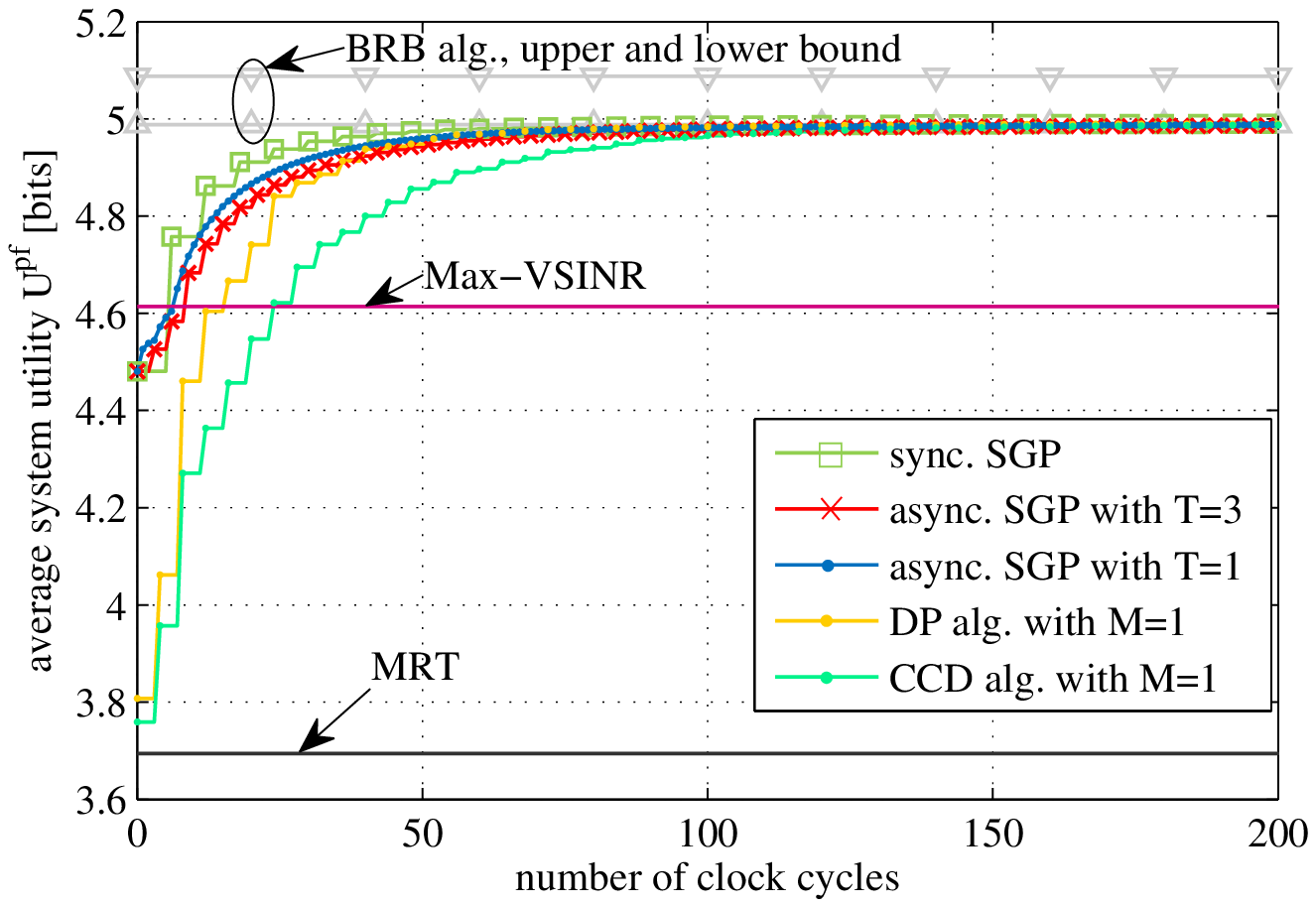}
\caption{Comparison of synchronous and asynchronous algorithms for NW1 with $D=3$. The plot shows the system utility $U^\mathrm{pf}$ as a function of the time index. The results are averaged over 100 channel realizations. }
\label{PXfig:bhnws1}
\end{figure}

\subsubsection{Comparison of Synchronous and Asynchronous Algorithms} 
Next, we compare the convergence rates of asynchronous and synchronous algorithms for different backhaul topologies. One should note that the parameter space of feasible algorithm and network configurations is very large. In the following, we concentrate on a set of parameters that best illustrates the principle behavior of the algorithms:
\begin{itemize}
	\item For each backhaul link, we assume the delay $D=3$.
	\item The asynchronous SGP algorithm is evaluated for two different update cycles $T=1$ and $T=3$. 
	\item The DP (resp. CCD) algorithm requires $M=1$ clock cycle for the 'measurement and reporting' task.	
\end{itemize}
As reference cases we plot the upper and lower bound for the optimal utility value, obtained by the BRB algorithm with $\epsilon=10^{-1}$, and the utility obtained by the maximum virtual SINR (Max-VSINR) beamforming algorithm from \cite{VSINR}.\\
Figures \ref{PXfig:bhnws1}-\ref{PXfig:bhnws3} illustrate the average utility $U^\mathrm{pf}$ as a function of the processing time. The results are averaged over 100 channel realizations. Numerical results are given in Table \ref{PX:tab:numCyclesConv}, which provides the mean and standard deviation of the number of (simulation) clock cycles needed for the algorithms' convergence. For each channel realization, the algorithms are run for a maximum number of $n_\mathrm{max}=10\,000$ clock cycles. The time instant of convergence is determined when the algorithm achieves 99\% of the utility $U^\mathrm{pf}(\mathbfit{X}(n_\mathrm{max}))$. 

\begin{figure}[t]
\centering
\includegraphics[width=0.5\textwidth]{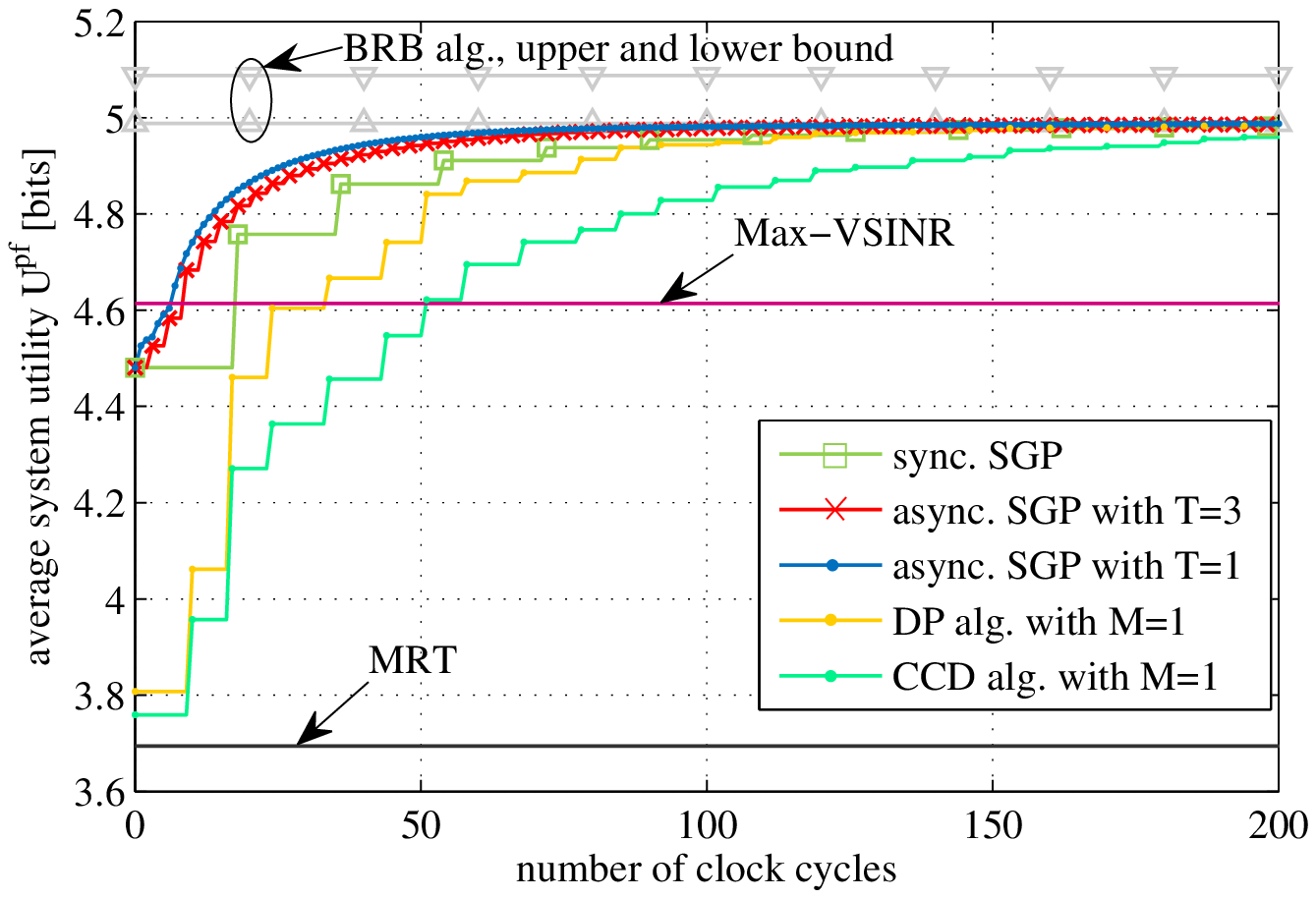}
\caption{Comparison of synchronous and asynchronous algorithms for NW2 with $D=3$. The plot shows the system utility $U^\mathrm{pf}$ as a function of the time index. The results are averaged over 100 channel realizations. }
\label{PXfig:bhnws2}
%
%
%
\includegraphics[width=0.5\textwidth]{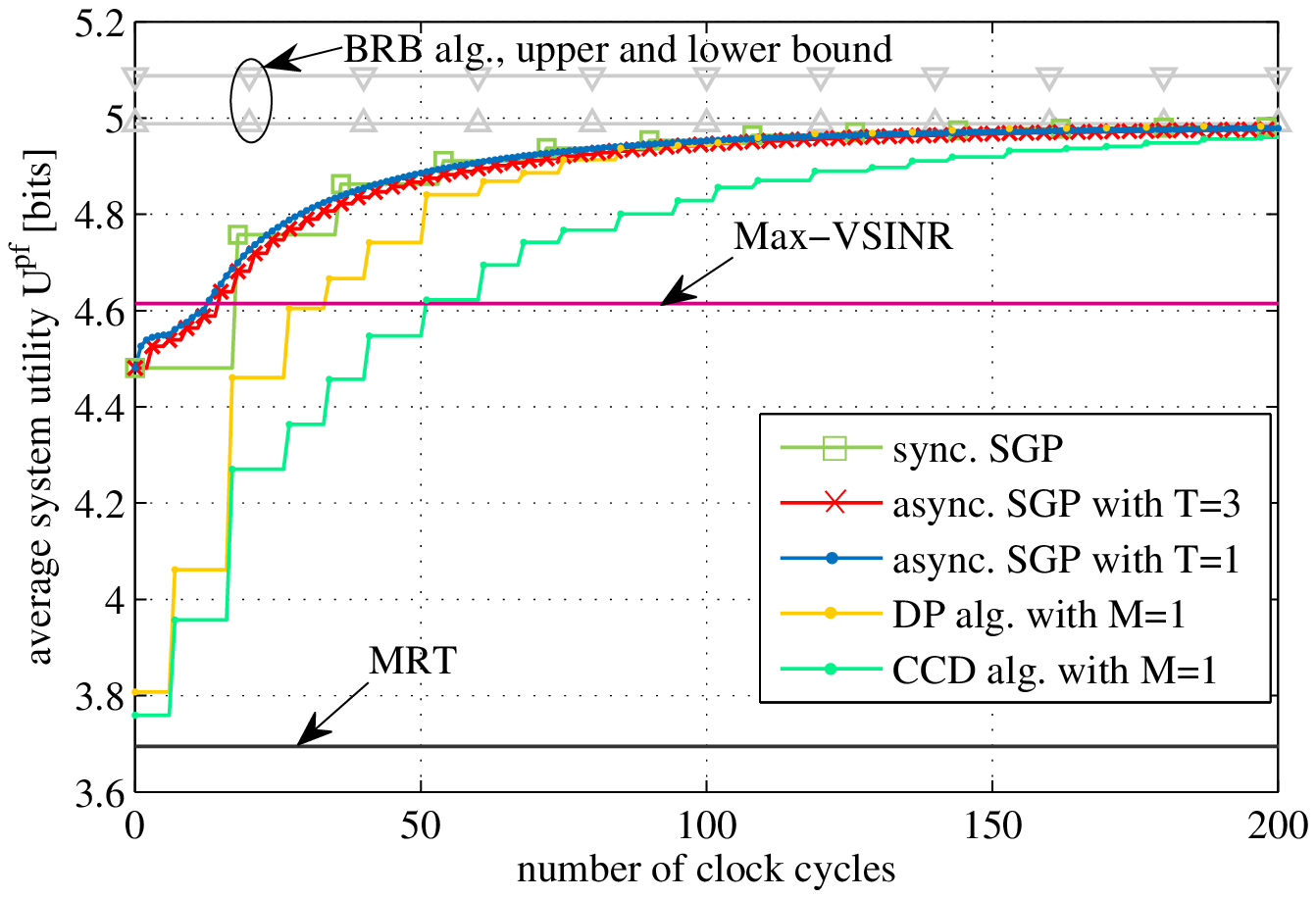}
\caption{Comparison of synchronous and asynchronous algorithms for NW3 with $D=3$. The plot shows the system utility $U^\mathrm{pf}$ as a function of the time index. The results are averaged over 100 channel realizations. }
\label{PXfig:bhnws3}
\end{figure}

\begin{table}[t]   
\small
\begin{tabular}{|l|l|l|l|l|l|}
    \hline
		             & sync. SGP     & \multicolumn{2}{c|}{async. SGP with}  & DP with          & CCD with \\
		             &               & $T=3$  &$T=1$   &   $M=1$    &  $M=1$ \\
		\hline
		\hline
		NW1 & $33.16$   & $56.26 $    & $44.86$       & $48.68$    &  $81.72$ \\
		      & $\pm 37.12$ & $\pm 64.78$    &  $\pm 50.17$    &  $\pm 36.13$ & $\pm 33.20$ \\
		\hline
		NW2 & $84.52$   & $56.80 $     & $45.44 $     & $102.30$ & $172.68$ \\
		      & $\pm 61.90$ & $\pm 66.56$    & $\pm 51.95$    & $\pm 76.07$&  $\pm 71.00$ \\
		\hline
		      & $84.52$   & $95.86$      & $71.96$      & $102.00 $   & $172.38$  \\		
		NW3 &  $\pm 61.90$ &$\pm 83.72$    & $\pm 44.58$    & $\pm 75.58$& $\pm 70.11$ \\		
			\hline	
\end{tabular}
\caption{Mean and standard deviation of the number of required clock cycles until convergence   \label{PX:tab:numCyclesConv}}
\vspace{-0.5cm}
\end{table}

Figure \ref{PXfig:bhnws1} shows the simulation results for the mesh backhaul network. First note that on average, all iterative algorithms converge to the same fixed point, which is close to the optimal solution as indicated by the lower bound of the BRB algorithm. By focusing on the convergence time, the synchronous SGP algorithm outperforms the asynchronous SGP algorithm variants. The reason is the loose bounding procedure for the gradient errors within the asynchronous SGP algorithm. However, as indicated in Figure \ref{PXfig:bhnws1}, decreasing the update interval $T$ of the asynchronous SGP algorithm reduces the convergence time. One can conclude that performing more frequent computations based on outdated information is beneficial. Moreover, the asynchronous SGP algorithm with $T=1$ outperforms the sequential DP and CCD algorithms, which primarily suffer from the synchronization penalty.\\
In Figures \ref{PXfig:bhnws2} and \ref{PXfig:bhnws3}, the simulation results for the two linear daisy chain networks are shown. Both networks induce diverse communication delays between the users. In NW2, these delays are matched to the channel gains, so that strongly coupled users are faced with small backhaul delays. In NW3, the backhaul links are permuted so that strongly coupled users observe large backhaul delays.
For the matched NW2, one can observe a clear benefit for the asynchronous SGP algorithm, which is able to exploit the fast backhaul links for exchanging substantial problem data. Here, the asynchronous SGP variants outperform all synchronous algorithms; that is, the convergence time is reduced by performing more frequent computations based on outdated information. Note that the relationship between the significance of problem data and the corresponding communication delays is critical. The NW3 possesses the same number of fast backhaul links, but the benefit of asynchronous computations is not present because the fast links carry insignificant data.

\subsubsection{Characterization of the SGP Algorithm's Limit Points} 
Finally, we check whether the limit points of the SGP algorithm are (local) maxima. Theoretically, each limit point can be a saddle point or a (local) minimum, although the latter case is very unlikely because the minima are repulsive. We evaluate 100 channel realizations. The SGP algorithm is run until convergence or a maximum number of iterations $n_\mathrm{max}=10\,000$ is reached. The convergence criterion is $\left\|\mathbfit{x}_k(n)-\mathbfit{x}_k(n-1)\right\|<10^{-6},\forall k$. After convergence, we use MATLAB function $\mathrm{fmincon}$ in order to maximize $U^\mathrm{pf}(\mathbfit{w}_1,\ldots,\mathbfit{w}_K)$, $\mathrm{s.t.}\left\|\mathbfit{w}_k\right\|\leq 1,\forall k$, starting from the limit point achieved by the SGP algorithm. The mean (resp. standard deviation) of the utility improvement is $4.231\cdot10^{-5}$[bits] (resp. $5.141\cdot10^{-5}$[bits]), which indicates for our simulations that the SGP algorithm always converges to a (local) maximum. Moreover, the corresponding correlation matrices are rank-one in all cases, which is shown by the mean and standard deviation of the smallest eigenvalues of the $2\times2 $ correlation matrices, given by $1.3596\cdot 10^{-7}$ and $6.2396\cdot 10^{-6}$, respectively.

\section{Conclusion}
The sum utility problem for the MISO IFC is re-parameterized in terms of power gains, which allows a condensed description of the user interactions. By adopting the scaled gradient projection (SGP) method, the users are allowed to perform linear update steps autonomously, based on possibly outdated gradient information. Assuming upper bounds on the objective's curvature as well as on the communication delays, we provide sufficient conditions for the asynchronous SGP algorithm that ensure the convergence to a stationary point. As illustrated by our numerical experiments, the derived step size bounds are not tight and thus yield slow convergence. However, we identify a class of utility functions (including the sum rate and proportional fair rate), for which the curvature bounds can be adjusted during the optimization process, while the algorithm's convergence behavior is preserved. Our simulations indicate a significant convergence speed-up for the resulting second-order algorithm. Finally, the convergence rate of different synchronous and asynchronous algorithms is compared. The main insight is that frequent asynchronous computations provide a convergence speed-up if the backhaul structure (in terms of communication delays) is matched to the coupling strength of the subproblems. This should normally be the case for mobile networks. However, the convergence speed-up comes at the cost of an increased number of update computations and exchanged messages.

\appendices

\section{Preliminaries: Convex Geometry and the Joint Numerical Range} \label{app_review}
In this section, we review some basic concepts from convex geometry \cite{hiriart2001fundamentals} that are utilized in Appendices \ref{prfParOpt} and \ref{prfUtilOpt}. We focus on the description of compact convex sets as the intersections of half-spaces, yielding an outer description of these sets. Then, we apply these concepts to the \emph{joint numerical range} \cite{Gutkin2004143}, which plays an essential role in our analysis because it is the generating set for the power gain region.

We begin with nonempty compact subsets $\mathcal{K}$ in $\mathbb{R}^K$. The outer boundary of $\mathcal{K}$ is denoted by $\partial_0 \mathcal{K}$, and is defined as the boundary between $\mathcal{K}$ and the unbounded component of $\mathbb{R}^K\backslash \mathcal{K}$. By $\mathrm{co}(\mathcal{K})$ we denote the convex hull of the set $\mathcal{K}$.

\begin{defn}[Partial Order on Vectors] \label{defDomVec}
Let $\mathbfit{x},\mathbfit{y}\in\mathbb{R}^K$. A vector $\mathbfit{y}$ \emph{dominates} a vector $\mathbfit{x}$ \emph{in direction} $\mathbfit{e}\in\left\{-1,+1\right\}^K$, written as $\mathbfit{y} \geq^{\mathbfit{e}} \mathbfit{x}$, if $y_le_l\geq x_le_l, \forall l$, and the inequality has at least one strict inequality.
\end{defn}

\begin{defn}[Outer Boundary Parts] \label{upperBoundary}
A point $\mathbfit{y}\in\mathbb{R}_+^K$ is called an \emph{outer boundary point} of a nonempty compact subset $\mathcal{K}\subset\mathbb{R}_+^K$ \emph{in direction} $\mathbfit{e}\in\left\{-1,+1\right\}^K$ if $\mathbfit{y}\in\mathcal{K}$ while the set $\left\{\mathbfit{y}'\in\mathbb{R}_+^K:\mathbfit{y}'\geq^\mathbfit{e} \mathbfit{y}\right\} \subset \mathbb{R}_+^K\backslash \mathcal{K}$. The set of all outer boundary points in direction $\mathbfit{e}$ is called the \emph{outer boundary part} of $\mathcal{K}$ \emph{in direction} $\mathbfit{e}$, and is denoted by $\partial_0^\mathbfit{e}\mathcal{K}$.
\end{defn}

Next, we introduce some basic concepts from convex geometry, that will help us to characterize convex sets. 

\begin{defn}[Support Function, Supporting Hyperplane / Halfspace]
Let $\mathcal{K}$ be a nonempty compact subset of $\mathbb{R}^K$ and $\boldsymbol\eta \in \mathbb{R}^K,\boldsymbol\eta\neq 0$. The function $s_{\mathcal{K}}(\boldsymbol\eta)=\max_{\mathbfit{y}\in\mathcal{K}} \boldsymbol\eta^T\mathbfit{y}$ is the \emph{support function} of $\mathcal{K}$ if it is convex and positive homogeneous (i.e., $s_{\mathcal{K}}(\alpha \boldsymbol\eta)=\alpha s_{\mathcal{K}}(\boldsymbol\eta)$ for $\alpha \in \mathbb{R}_+$).\\
The \emph{supporting hyperplane} (resp. \emph{halfspace}) of $\mathcal{K}$ in direction $\boldsymbol\eta$ is given by
$\mathcal{H}\left(\boldsymbol\eta,s_{\mathcal{K}}(\boldsymbol\eta)\right)=\left\{\mathbfit{y}\in\mathbb{R}^K:\boldsymbol\eta^T\mathbfit{y}=s_{\mathcal{K}}(\boldsymbol\eta)\right\}$  (resp. $\mathcal{H}^-\left(\boldsymbol\eta,s_{\mathcal{K}}(\boldsymbol\eta)\right)=\left\{\mathbfit{y}\in\mathbb{R}^K:\boldsymbol\eta^T\mathbfit{y}\leq s_{\mathcal{K}}(\boldsymbol\eta)\right\}$).
\end{defn}
Due to the positive homogeneity, the support function is completely determined by its value on the unit sphere $\mathcal{S}^{K-1}$. Consequently, for $\boldsymbol \eta \in S^{K-1}$, $s_{\mathcal{K}}(\boldsymbol\eta)$ is the signed distance of $\mathcal{H}\left(\boldsymbol\eta,s_{\mathcal{K}}(\boldsymbol\eta)\right)$ from the origin.\\
A fundamental concept in convex geometry is the outer description of convex sets. Every nonempty compact convex set $\mathcal{C}=\mathrm{co}(\mathcal{K})$ is given by the intersection of its supporting halfspaces (\cite[Theorem 2.2.2]{hiriart2001fundamentals}; that is,
\begin{align}
\mathcal{C}&=\bigcap_{\boldsymbol\eta\in\mathcal{S}^{K-1}} \mathcal{H}^- \left(\boldsymbol\eta,s_{\mathcal{K}}(\boldsymbol\eta)\right)\\
&= \left\{\mathbfit{y}\in\mathbb{R}^K\left|\right.\boldsymbol\eta^T\mathbfit{y}\leq s_{\mathcal{K}}(\boldsymbol\eta):\boldsymbol\eta \in \mathcal{S}^{K-1}\right\}.
\end{align}
Note that the support function determines the set $\mathcal{C}$ uniquely. It can be used to describe certain geometric properties of convex sets analytically.

Next, we describe specific parts of the (outer) boundary of $\mathcal{K}$, which are determined by the surface normal vector $\boldsymbol\eta$.
\begin{defn}[Exposed Face]
The \emph{exposed face} of $\mathcal{K}$ with the surface normal $\boldsymbol \eta \in S^{K-1}$ is given by the support set
\begin{align}
\Phi_{\mathcal{K}}(\boldsymbol\eta)=\mathcal{K} \cap \mathcal{H}\left(\boldsymbol\eta,s_{\mathcal{K}}(\boldsymbol\eta)\right).
\end{align}
\end{defn}

By \cite[Proposition 3.1]{Gutkin2004143} we have $\partial\mathrm{co}(\mathcal{K})=\partial_0\mathcal{K}$ if and only if $\Phi_{\mathcal{K}}(\boldsymbol\eta)$ is convex for any $\boldsymbol\eta \in \mathcal{S}^{K-1}$.\\

We now present a few general results pertaining to the joint numerical range and its convex hull. Moreover, we illustrate its connection to the power gain region defined in \eqref{eq:pg_region:def}.
\begin{defn}[Joint Numerical Range]
Let $\mathbfit{H}=(\mathbfit{H}_1,\ldots,\mathbfit{H}_K)^H$ be a $K$-tuple of Hermitian matrices with $\mathbfit{H}_l \in \mathbb{C}^{N\times N},\forall l$. The \emph{joint numerical range} (JNR) of the matrices $\mathbfit{H}_1,\ldots,\mathbfit{H}_K$ is defined as
\begin{align}
\mathcal{F}(\mathbfit{H})=&\left\{\left(\mathbfit{w}^H\mathbfit{H}_1\mathbfit{w},\ldots, \mathbfit{w}^H\mathbfit{H}_K\mathbfit{w}  \right)^T:\right.\nonumber\\
&\left. \mathbfit{w}\in\mathbb{C}^N,\left\|\mathbfit{w}\right\|=1\right\}.
\end{align}
\end{defn}
This set is compact but for $K>2$ not necessarily convex. For $K\leq 3$, the outer boundary of $\mathcal{F}(\mathbfit{H})$ is convex. For recent studies and developments concerning the (lack of) convexity for the joint numerical range see, e.g., \cite{Gutkin2004143},\cite{PoonJNR}.\\
If $\mathcal{V} \subset \mathbb{C}^N$ is a subspace then the joint numerical range of the restriction of $\mathbfit{H}$ to $\mathcal{V}$ is denoted by  
\begin{align}
\mathcal{F}\left(\mathbfit{H};\mathcal{V}\right)=&\left\{\left(\mathbfit{w}^H\mathbfit{H}_1\mathbfit{w},\ldots, \mathbfit{w}^H\mathbfit{H}_K\mathbfit{w}  \right)^T:\right.\nonumber\\
 &\left.\mathbfit{w}\in\mathcal{V},\left\|\mathbfit{w}\right\|=1\right\}.\nonumber
\end{align}

The convex hull of the set $\mathcal{F}(\mathbfit{H})$ is referred to as the \emph{joint field of values} (JFV) \cite{barker}, and is given by 
\begin{align}
\mathcal{W}(\mathbfit{H})=&\mathrm{co}\left(\mathcal{F}(\mathbfit{H})\right) \label{eq_WcoF}\\
=&\left\{\left(\mathrm{tr}(\mathbfit{Q}\mathbfit{H}_{1}),\ldots, \mathrm{tr}(\mathbfit{Q}\mathbfit{H}_{K})\right)^T:\right.\nonumber\\
&\left. \mathbfit{Q} \in \mathbb{C}^{N\times N},\mathrm{tr}\left(\mathbfit{Q}\right) = 1, \mathbfit{Q}\succeq 0\right\}.
\end{align}
By rewriting $\mathbfit{w}^H\mathbfit{H}_k\mathbfit{w}=\mathrm{tr}(\mathbfit{w}\mathbfit{w}^H\mathbfit{H}_k),\forall k$, the difference between $\mathcal{F}(\mathbfit{H})$ and $\mathcal{W}(\mathbfit{H})$ can be easily observed: The JNR is generated by Hermitian rank-one matrices only, while the JFV is obtained by using Hermitian matrices of arbitrary rank. 

Finally, the set $\mathcal{W}(\mathbfit{H})$ can be used to define a generalization of the power gain region $\Omega_k$ for $K$-tuples of arbitrary Hermitian matrices, which yields
\begin{align}
\Omega(\mathbfit{H})=&\mathrm{co}\left(\boldsymbol{0}\cup\mathcal{W}(\mathbfit{H})\right)\label{eq_OcoW0}\\
=&\left\{\left(\mathrm{tr}(\mathbfit{Q}\mathbfit{H}_{1}),\ldots, \mathrm{tr}(\mathbfit{Q}\mathbfit{H}_{K})\right)^T: \mathbfit{Q} \in \mathcal{Q}\right\}
\end{align}
where $\mathcal{Q}=\left\{\mathbfit{Q} \in \mathbb{C}^{N \times N}, \mathrm{tr}(\mathbfit{Q})\leq 1, \mathbfit{Q}\succeq 0\right\}$.
By setting $\mathbfit{H}=(\mathbfit{h}_{k1}\mathbfit{h}_{k1}^H,\ldots,\mathbfit{h}_{kK}\mathbfit{h}_{kK}^H)^H$, we obtain the $k$-th power gain region $\Omega_k=\Omega(\mathbfit{H})$. The relationship between $\Omega_k$ and $\mathcal{F}(\mathbfit{H})$ is illustrated in Fig. \ref{fig:jnr_pgregion} for the power gain region of user $1$ in the two-user case.

\begin{figure}[t]
\centering
\includegraphics[width=7.0cm]{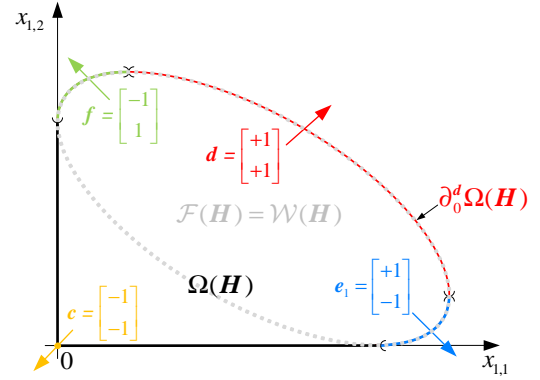}
\caption{Illustration of a two-dimensional joint numerical range $\mathcal{F}(\mathbfit{H})$ (dotted gray curve) and the corresponding power gain region $\Omega(\mathbfit{H})$ (solid curve) with its outer boundary parts (e.g., red curve $\partial_0^{\mathbfit{d}}\Omega(\mathbfit{H})$). Since $K=2$, we have $\mathcal{F}(\mathbfit{H})=\mathcal{W}(\mathbfit{H})$; that is, the joint numerical range is convex.}
\label{fig:jnr_pgregion}
\end{figure}

\begin{prop}[Properties of the JNR/JFV/Power Gain Region]\label{prop1} The following claims hold:
\begin{itemize}
  \item[(i)] The support function for the set $\mathcal{F}(\mathbfit{H})$ is given by\footnote{Note that $\boldsymbol{\eta}^T\mathbfit{H}$ represents the 'inner product' of a $K$-dimensional vector $\boldsymbol{\eta}$ with the $K$-tuple $\mathbfit{H}=(\mathbfit{H}_1,\ldots,\mathbfit{H}_K)^H$; i.e., $\boldsymbol{\eta}^T\mathbfit{H}=\sum_{l=1}^K \eta_l \mathbfit{H}_l$.} $s_{\mathcal{F}(\mathbfit{H})}(\boldsymbol\eta)=\lambda_1(\boldsymbol{\eta}^T\mathbfit{H})$.
	\item[(ii)] The subset of outer boundary points of $\mathcal{F}(\mathbfit{H})$ in direction $\boldsymbol\eta$ is given by $\Phi_{\mathcal{F}(\mathbfit{H})}(\boldsymbol\eta)=\mathcal{F}\left(\mathbfit{H};\mathcal{E}_1(\boldsymbol{\eta}^T\mathbfit{H})\right)$. If $\mathrm{dim}\left(\mathcal{E}_1(\boldsymbol{\eta}^T\mathbfit{H})\right)=1$ then $\Phi_{\mathcal{F}(\mathbfit{H})}(\boldsymbol\eta)$ is a singleton.
	\item[(iii)] The sets $\mathcal{F}(\mathbfit{H})$ and $\mathcal{W}(\mathbfit{H})$ share the same support function; that is, $s_{\mathcal{W}(\mathbfit{H})}(\boldsymbol\eta)=s_{\mathcal{F}(\mathbfit{H})}(\boldsymbol\eta)$.
	\item[(iv)] We have
$\Phi_{\mathcal{W}(\mathbfit{H})}(\boldsymbol\eta) = \mathrm{co}\left(\Phi_{\mathcal{F}(\mathbfit{H})}(\boldsymbol\eta)\right)$.
\end{itemize}
For the last two claims, we further assume $\mathbfit{H}=(\mathbfit{h}_{k1}\mathbfit{h}_{k1}^H,\ldots,\mathbfit{h}_{kK}\mathbfit{h}_{kK}^H)^H$:
\begin{itemize}
  \item[(v)] For all $\boldsymbol \eta \in S^{K-1}$, we have $s_{\Omega(\mathbfit{H})}(\boldsymbol\eta)=\mathrm{max}\left(0,\lambda_1(\boldsymbol{\eta}^T\mathbfit{H})\right) \geq 0$ .
	\item[(vi)] For $s_{\Omega(\mathbfit{H})}(\boldsymbol\eta)>0$, we have $\Phi_{\Omega(\mathbfit{H})}(\boldsymbol\eta)=\Phi_{\mathcal{W}(\mathbfit{H})}(\boldsymbol\eta)$.
\end{itemize}
\end{prop}

\begin{proof} 
The claims (i) and (ii) are given by \cite[Proposition 3.5]{Gutkin2004143}. The third claim follows from relation \eqref{eq_WcoF}. Claim (iv) follows from \cite[Proposition 3.1]{Gutkin2004143} and relation \eqref{eq_WcoF}. The fifth claim is immediate from claim (iii) and relation \eqref{eq_OcoW0}; that is, $s_{\Omega(\mathbfit{H})}(\boldsymbol\eta)=\mathrm{max}(0,s_{\mathcal{W}(\mathbfit{H})}(\boldsymbol\eta))= \mathrm{max}\left(0,\lambda_1(\boldsymbol{\eta}^T\mathbfit{H})\right)\geq 0$. Claim (vi) follows from claim (v).
\end{proof}

\section{Proof of Theorem \ref{thm1}}\label{prfParOpt}
The proof that single-stream beamforming is sufficient to achieve all Pareto optimal points is accomplished in two steps. The first part is identical to \cite[Theorem 2]{5643183}, \cite[Lemma 1.5]{CIT-069}. There, it is shown that the Pareto boundary $\mathcal{PB}(\mathcal{U})$ is achieved by transmit correlation matrices $\mathbfit{Q}_1,\ldots,\mathbfit{Q}_K$ that, for each $k$, also achieve the outer boundary part of the power gain region $\Omega_k$ in the direction $\mathbfit{e}_k=\left[-1 \ldots -1 +1 -1 \ldots -1\right]^T$, where only the $k$-th component is positive. The proof works by contradiction. Assume that $\mathbfit{Q}_1,\ldots,\mathbfit{Q}_K$ achieve a point on the Pareto boundary $\mathcal{PB}(\mathcal{U})$ but there is a user $k$ whose power gain vector $\mathbfit{x}_k(\mathbfit{Q}_k)$ is not on the outer boundary part $\partial_0^\mathbfit{e_k}\Omega_k$. Then, it is possible to increase the $k$-th component of $\mathbfit{x}_k(\mathbfit{Q}_k)$ without changing the other components. By Assumption \ref{ukProp} on the monotonicity of $u_k$, the modified power gain vector leads to an improved utility of user $k$ and unchanged utilities for all other users. This is a contradiction to the assumption that $\mathbfit{Q}_1,\ldots,\mathbfit{Q}_K$ achieved the Pareto boundary $\mathcal{PB}(\mathcal{U})$.

In the second part of the proof, we show that \emph{all} boundary points in $\partial_0^\mathbfit{e_k}\Omega_k$  can be achieved by correlation matrices with $\mathrm{rank}(\mathbfit{Q}_k)\leq 1$. By symmetry, it suffices to consider the $k$-th user only. For our analysis, we adopt the methods from convex geometry introduced in Appendix \ref{app_review}. We begin with a review of the solution approach from \cite[Lemma 3]{5643183} and \cite[Lemma 1.7]{CIT-069}, in which the problem is examined for some \emph{arbitrary} outer boundary part; that is, $\mathbfit{e}\in \left\{-1,+1\right\}^K$. Based on the Supporting Hyperplane Theorem \cite[Theorem 1.5]{tuy1998convex}, the authors in \cite{5643183},\cite{CIT-069} characterize every exposed face $\Phi_{\Omega_k}(\boldsymbol\eta)$ of $\Omega_k$ (with normal vector $\boldsymbol\eta\in \mathcal{S}^{K-1}$) by the following optimization problem
\begin{align}
&\max_{\mathbfit{Q}_k\in\mathcal{Q}} \boldsymbol\eta^T \mathbfit{x}_k(\mathbfit{Q}_k)\;\;\;\;\mathrm{s.\,t.}\; \mathrm{tr}(\mathbfit{Q}_k)\leq 1.   \tag{$\mathrm{P}_4$} \label{expFaceProb}\nonumber
\end{align}
They show that Problem \eqref{expFaceProb} has always solutions with $\mathrm{rank}(\mathbfit{Q}_k)\leq 1$; that is, there always exists a point $\mathbfit{y}\in \Phi_{\Omega_k}(\boldsymbol\eta)$ that is achieved by correlation matrix with $\mathrm{rank}(\mathbfit{Q}_k)\leq 1$. 
Unfortunately, the set $\Phi_{\Omega_k}(\boldsymbol\eta)$ is not necessarily a singleton (i.e., Problem \eqref{expFaceProb} may have several solutions), which is the case when there exist multiple points on the outer boundary  $\partial_0\Omega_k$ with the same normal vector $\boldsymbol\eta$. 
In order to complete the proof of \cite[Lemma 3]{5643183} and \cite[Lemma 1.7]{CIT-069}, it must be shown that \emph{all} elements of $\Phi_{\Omega_k}(\boldsymbol\eta)$ can be achieved by correlation matrices with $\mathrm{rank}(\mathbfit{Q}_k)\leq 1$. 

We briefly illustrate the difficulty of this problem. By Proposition \ref{prop1} (v), the optimal value of Problem \eqref{expFaceProb} is given by the support function $s_{\Omega_k}(\boldsymbol\eta)\geq 0$. Now, consider an exposed face $\Phi_{\Omega_k}(\boldsymbol\eta)$ of $\Omega_k$ with $s_{\Omega_k}(\boldsymbol\eta)> 0$. By Proposition \ref{prop1} (vi) and (iv), the exposed face can be written as $\Phi_{\Omega_k}(\boldsymbol\eta)=\Phi_{\Omega(\mathbfit{H})}(\boldsymbol\eta)=\Phi_{\mathcal{W}(\mathbfit{H})}(\boldsymbol\eta)=\mathrm{co}\left(\Phi_{\mathcal{F}(\mathbfit{H})}(\boldsymbol\eta)\right)$ with $\mathbfit{H}=(\mathbfit{h}_{k1}\mathbfit{h}_{k1}^H,\ldots,\mathbfit{h}_{kK}\mathbfit{h}_{kK}^H)^H$. This means that we have to show that the set $\Phi_{\mathcal{F}(\mathbfit{H})}(\boldsymbol\eta)$ is convex. By Proposition \ref{prop1} (ii), this set is given by $\Phi_{\mathcal{F}(\mathbfit{H})}(\boldsymbol\eta)=\mathcal{F}\left(\mathbfit{H};\mathcal{E}_1(\boldsymbol{\eta}^T\mathbfit{H})\right)$; that is, the exposed face is itself a joint numerical range. Since none of the known conditions for convexity of the joint numerical range (see, e.g., \cite{PoonJNR}) applies for the general case with arbitrary $N$ and $K$, the problem as treated in \cite{5643183},\cite{CIT-069} remains \emph{unsolved}.

However, in order to prove the sufficiency of single-stream beamforming for Pareto optimality, it suffices to consider only the outer boundary part $\partial_0^\mathbfit{e_k}\Omega_k$. As illustrated in Fig. \ref{fig:jnr_pgregion}, this boundary part corresponds to the set of exposed faces $\Phi_{\Omega_k}(\boldsymbol\eta)$ with the normal vectors $\boldsymbol\eta\in \mathcal{T}_k$, where $\mathcal{T}_{k}=\left\{\boldsymbol\eta \in \mathcal{S}^{K-1}: \eta_k > 0, \eta_l < 0,\forall l\neq k\right\}$. The idea behind our proof is to distinguish between exposed faces on the conical boundary part of $\Omega_k$ (i.e., all sets $\Phi_{\Omega_k}(\boldsymbol\eta)$ with $s_{\Omega_k}(\boldsymbol\eta)= 0$), and exposed faces with $s_{\Omega_k}(\boldsymbol\eta)>0$. For the latter set we show that if $\boldsymbol\eta\in \mathcal{T}_k$ then $\Phi_{\Omega_k}(\boldsymbol\eta)$ is always a singleton and thus convex. We then complete the proof by showing that the exposed faces on the conical boundary part with $\boldsymbol\eta\in \mathcal{T}_k$ can be reached by scaled versions of the exposed faces that are singletons. The formal proof is given by Proposition \ref{prop3893883} with $\Omega_k=\Omega(\mathbfit{H})$.

\begin{prop} \label{prop3893883}
If $\boldsymbol\eta \in \mathcal{T}_{k}$ then all points in $\Phi_{\Omega(\mathbfit{H})}(\boldsymbol\eta)$ can be achieved by correlation matrices with $\mathrm{rank}(\mathbfit{Q}_k)\leq 1$.
\end{prop}
\begin{proof}
By Proposition \ref{prop1} (v) the support function of $\Omega(\mathbfit{H})$ is always non-negative; that is, $s_{\Omega(\mathbfit{H})}(\boldsymbol{\eta})\geq 0$, $\forall \boldsymbol\eta\in\mathcal{S}^{K-1}$. We distinguish between the two cases:\\
1) If $s_{\Omega(\mathbfit{H})}(\boldsymbol{\eta})> 0$ then by Proposition \ref{prop1} (claims (ii), (iv), (vi)) we have $\Phi_{\Omega(\mathbfit{H})}(\boldsymbol\eta)=\mathrm{co}\left(\mathcal{F}(\mathbfit{H};\mathcal{E}_1(\boldsymbol{\eta}^T\mathbfit{H})\right)$. By showing that the set $\mathcal{F}(\mathbfit{H};\mathcal{E}_1(\boldsymbol{\eta}^T\mathbfit{H}))$ is a singleton, we ensure that the exposed face $\Phi_{\Omega(\mathbfit{H})}(\boldsymbol\eta)$ is achieved by a rank-one correlation matrix. Thus, we only have to prove that the eigenspace $\mathcal{E}_1(\boldsymbol{\eta}^T\mathbfit{H})$ has dimension one; that is, the geometric multiplicity of the largest eigenvalue $\lambda_1(\boldsymbol\eta^T\mathbfit{H})$ equals one. This is accomplished by showing that the first and second eigenvalue of the matrix $\boldsymbol\eta^T\mathbfit{H}$ are strictly separated. Therefore, we rewrite $\boldsymbol\eta^T\mathbfit{H} = \mathbfit{A} + \mathbfit{B}$ with $\mathbfit{A}=\eta_k \mathbfit{h}_{k}\mathbfit{h}_{k}^H$ and $\mathbfit{B}=\sum_{l\neq k} \eta_l \mathbfit{h}_{l}\mathbfit{h}_{l}^H$. 
If $\boldsymbol{\eta} \in \mathcal{T}_k$ and $s_{\Omega(\mathbfit{H})}(\boldsymbol{\eta})>0$, then we have $\eta_k > 0, \eta_l < 0,\forall l\neq k$  and $\mathbfit{A} \succeq 0$, $\mathrm{rank}(\mathbfit{A})\leq 1$, $\mathbfit{B} \preceq 0$. Applying Weyl's eigenvalue inequality \cite[Section 1.3]{taotopics} yields 
\begin{align}
\lambda_1(\boldsymbol\eta^T\mathbfit{H})-\lambda_2(\boldsymbol\eta^T\mathbfit{H}) &\geq \lambda_1(\boldsymbol\eta^T\mathbfit{H}) - \lambda_2(\mathbfit{A})-\lambda_1(\mathbfit{B})\nonumber\\
& = \lambda_1(\boldsymbol\eta^T\mathbfit{H}) + \left|\lambda_1(\mathbfit{B})\right|\nonumber\\
& >0.
\end{align}
Hence, $\mathrm{dim}\left(\mathcal{E}_1(\boldsymbol\eta^T\mathbfit{H})\right)=1$; that is, the set $\Phi_{\Omega(\mathbfit{H})}(\boldsymbol\eta)=\mathcal{F}(\mathbfit{H};\mathcal{E}_1(\boldsymbol\eta^T\mathbfit{H}))$ is a singleton.\\
2) For the case $s_{\Omega(\mathbfit{H})}(\boldsymbol{\eta}) = 0$, we consider a certain point $\mathbfit{y}\in\Phi_{\Omega(\mathbfit{H})}(\boldsymbol\eta)$ with $\boldsymbol\eta\in\mathcal{T}_k$. 
We show that every neighborhood of $\mathbfit{y}$ contains a point that is achieved by a correlation matrix with $\mathrm{rank}(\mathbfit{Q}_k)\leq 1$; that is, $\mathbfit{y}$ is a limit point of a (scaled) joint numerical range. Since a closed set contains its limit points, the point $\mathbfit{y}$ is likewise achieved by a correlation matrix with $\mathrm{rank}(\mathbfit{Q}_k)\leq 1$.\\
Let $\rho\in[0,1]$ be the smallest scaling factor such that $\mathbfit{y}\in \Omega(\rho\mathbfit{H})$. The boundary of $\Omega(\rho\mathbfit{H})$ can be divided into two (possibly overlapping) sets $\mathcal{A}$ and $\mathcal{B}$, with $\partial\Omega(\rho\mathbfit{H})=\mathcal{A}\cup \mathcal{B}$. The conical boundary part of $\Omega(\rho\mathbfit{H})$ is given by the closed set
\begin{align}
 \mathcal{A}=\left\{\mathbfit{y}\in \Phi_{\Omega(\rho\mathbfit{H})}(\boldsymbol\eta):\boldsymbol\eta\in\mathcal{S}^{K-1},s_{\Omega(\rho\mathbfit{H})}(\boldsymbol\eta)=0\right\},\nonumber
\end{align}
while the remaining boundary part is included in the set
\begin{align}
 \mathcal{B}=\left\{\mathbfit{y}\in \Phi_{\Omega(\rho\mathbfit{H})}(\boldsymbol\eta):\boldsymbol\eta\in\mathcal{S}^{K-1},s_{\Omega(\rho\mathbfit{H})}(\boldsymbol\eta)>0\right\}.\nonumber
\end{align}
By definition of $\rho$, the point $\mathbfit{y}$ must lie on the boundary of the subset $\mathcal{A}$. Consequently, every open neighborhood of $\mathbfit{y}$ contains at least one point $\mathbfit{y}'\in\mathcal{B}$ with corresponding $\boldsymbol{\eta}'$ and $s_{\Omega(\rho\mathbfit{H})}(\boldsymbol{\eta}')>0$. Note that $\boldsymbol{\eta}'\in\mathcal{T}_k$ because $\partial\Omega(\rho\mathbfit{H})$ is a (continuous) convex curve and $\mathcal{T}_k$ is an open set.\\
By applying case 1), we have $\mathbfit{y}'\in\mathcal{F}(\rho\mathbfit{H})$. If $\rho>0$ then every such point $\mathbfit{y}'$ can be achieved by a rank-one correlation matrix. Since every neighborhood of $\mathbfit{y}$ contains such a point $\mathbfit{y}'$, the point $\mathbfit{y}$ is a limit point of the (closed) set $\mathcal{F}(\rho\mathbfit{H})$ and thus must be itself an element of this set. If $\rho=0$, then we simply have $\mathrm{rank}(\mathbfit{Q}_k)=0$.
\end{proof}

\section{Proof of Theorem \ref{thm2}}\label{prfUtilOpt}
The proof that all stationary points of \eqref{bfOptProb} can be achieved with single-stream beamforming is based on the convex geometry of the power gain region, see Appendix \ref{app_review}. By symmetry, it suffices to consider the $k$-th user only. Set $\mathbfit{H}=(\mathbfit{h}_{k1}\mathbfit{h}_{k1}^H,\ldots,\mathbfit{h}_{kK}\mathbfit{h}_{kK}^H)^H$, then we have $\Omega_k=\Omega(\mathbfit{H})$. The Condition \eqref{eq:optcond1} can be reformulated as
$\mathbfit{x}_k^* \in \Phi_{\Omega(\mathbfit{H})}(\boldsymbol\eta^*),\forall k$ 
where $\boldsymbol\eta^*=\slfrac{\nabla_{k}U\left(\mathbfit{X}^*\right)}{\left\|\nabla_{k}U\left(\mathbfit{X}^*\right)\right\|}$. 
By Assumption \ref{ukProp}, we have $\slfrac{\partial{U(\mathbfit{X}^*)}}{\partial{x_{k,l}}}<0,\forall l\neq k$ and $\slfrac{\partial{U(\mathbfit{X}^*)}}{\partial{x_{k,k}}}>0$. Consequently, the normal vector $\boldsymbol\eta^*$ must be an element of the set $\mathcal{T}_{k}=\left\{\boldsymbol\eta \in \mathcal{S}^{K-1}: \eta_k > 0, \eta_l < 0,\forall l\neq k\right\}$. Now, we can invoke Proposition \ref{prop3893883} which shows that all outer boundary points $\mathbfit{x}\in  \Phi_{\Omega(\mathbfit{H})}(\boldsymbol\eta)$ with $\boldsymbol\eta \in \mathcal{T}_{k}$ can be achieved by correlation matrices $\mathbfit{Q}_k$ with $\mathrm{rank}(\mathbfit{Q}_k)\leq 1$.\\
Next, we show how to find the corresponding beamforming vector $\mathbfit{w}_k^*$. By applying the scaled gradient projection algorithm (Section \ref{ADBF}), we obtain a correlation matrix $\mathbfit{Q}_k^*$ that achieves the $k$-th component $\mathbfit{x}_k^*$ of the stationary solution $\mathbfit{X}^*$. Depending on the rank of this matrix, we distinguish between the following two cases:
\begin{enumerate}
	\item  If $\mathrm{rank}(\mathbfit{Q}_k^*)\leq 1$ then we have $\mathbfit{Q}_k^*=\mathbfit{w}_k^*(\mathbfit{w}_k^*)^H$. The vector $\mathbfit{w}_k^*$ is given by the dominant eigenvector $v_\mathrm{max}(\mathbfit{Q}_k^*)$, scaled by the square root of the largest eigenvalue $\lambda_1(\mathbfit{Q}_k^*)$.
	\item If $\mathrm{rank}(\mathbfit{Q}_k^*)>1$ then we have to find a beamforming vector that achieves the power gain vector $\mathbfit{x}_k^*=\mathbfit{x}_k(\mathbfit{Q}_k^*)$, which yields the feasibility problem
\begin{align}
\mathrm{find}\;&\mathbfit{w}_k  \tag{$\mathrm{P}_5$} \label{findBfVecRank1Prob} \nonumber  \\
\mathrm{s.\,t.}\;  &\left|\mathbfit{h}_{kl}^H\mathbfit{w}_k\right|^2 = x_{k,l}(\mathbfit{Q}_k^*),\forall l\nonumber\\
&\left\|\mathbfit{w}_k\right\|^2\leq 1.\nonumber
\end{align}
This problem is non-convex due to the quadratic equality constraints. However, we can transform Problem \eqref{findBfVecRank1Prob} into a convex optimization problem by searching for beamforming vectors that yield a sum utility which is at least as good as the original one. By the monotonicity Assumption \ref{Uprop}, we can replace the equality constraints for the interference powers $x_{k,l}(\mathbfit{Q}_k^*),\forall l\neq k$, by the inequality constraints $\left|\mathbfit{h}_{kl}^H\mathbfit{w}_k\right|^2\leq x_{k,l}(\mathbfit{Q}_k^*),\forall l\neq k$. By maximizing the useful signal power $\left|\mathbfit{h}_{kk}^H\mathbfit{w}_k\right|^2$, we obtain the interference-constrained beamforming problem
\begin{align}
\min_{\mathbfit{w}_k} &-\left|\mathbfit{h}_{kk}^H\mathbfit{w}_k\right|^2 \tag{$\mathrm{P}_6$} \label{findBfVecRank1Prob2}  \nonumber \\
\mathrm{s.\,t.}\;  &\left|\mathbfit{h}_{kl}^H\mathbfit{w}_k\right|^2\leq x_{k,l}(\mathbfit{Q}_k^*),\forall l\neq k\nonumber\\
&\left\|\mathbfit{w}_k\right\|^2\leq 1.\nonumber
\end{align}
This problem is still non-convex due to the concave objective. Similar to \cite{4599181}, we note that any solution of Problem \eqref{findBfVecRank1Prob2} is invariant with respect to a phase rotation. Thus, the optimal solution can be found by assuming that $\mathbfit{h}_{kk}^H\mathbfit{w}_k$ is real and nonnegative, yielding the convex optimization problem in \eqref{findBfVecProb}. Note that \eqref{findBfVecRank1Prob}, \eqref{findBfVecRank1Prob2}, \eqref{findBfVecProb} are always feasible because the first part of Theorem \ref{thm2} proves the existence of a non-zero beamforming vector $\mathbfit{w}_k^*$ with $\left|\mathbfit{h}_{kl}^H\mathbfit{w}_k^*\right|^2= x_{k,l}(\mathbfit{Q}_k^*),\forall l$.
\end{enumerate}

\section{Proof of Theorem \ref{thm:tsit}} \label{proof:tsit}
Before proving Theorem \ref{thm:tsit}, we first establish a block-ascent property (cf. \cite[Lemma 5.1]{Bertsekas:1989:PDC}) for the scaled gradient projection algorithm. Therefore, we rewrite \eqref{eq:updateStep} as
\begin{align}
\mathbfit{x}_k(n+1)&=\left[\mathbfit{x}_k(n)+\gamma_k \mathbfit{M}_k^{-1} \boldsymbol{\lambda}_k(n)\right]_{\mathbfit{M}_k}^{\Omega_k}\nonumber\\
&=\mathbfit{x}_k(n)+\gamma_k \mathbfit{s}_k(n) \label{eq_x_s}
\end{align}
with $\mathbfit{s}_k(n)=1/\gamma_k([\mathbfit{x}_k(n)+\gamma_k \mathbfit{M}_k^{-1} \boldsymbol{\lambda}_k(n)]_{\mathbfit{M}_k}^{\Omega_k}-\mathbfit{x}_k(n))$.
\begin{lem}[Block Ascent Property]\label{lemma2}
Let $\mathbfit{M}_k=\mathrm{diag}\left(\beta_{k,1},\ldots,\beta_{k,K}\right)$, then we have for each $k$ and $n$
\begin{align}
\mathbfit{s}_k(n)^T\boldsymbol{\lambda}_k(n) \geq \mathbfit{s}_k(n)^T\mathbfit{M}_k \mathbfit{s}_k(n)=\sum_{l=1}^K\beta_{k,l} \left|s_{k,l}(n)\right|^2. \label{ieq_bap}
\end{align}
\end{lem}
\begin{proof}
By the Scaled Projection Theorem \cite[Proposition 3.7 (b)]{Bertsekas:1989:PDC} we have
\begin{align}
\left[\mathbfit{x}_k(n+1)-\mathbfit{x}_k(n)\right]^T\mathbfit{M}_k \left[\mathbfit{x}_k(n+1)-\right.\nonumber\\
\left.\left(\mathbfit{x}_k(n)+\gamma_k \mathbfit{M}_k^{-1} \boldsymbol{\lambda}_k(n)\right)\right] &\leq 0\nonumber
\end{align}
Equivalently, we can write
\begin{align}
\gamma_k \mathbfit{s}_k(n)^T \mathbfit{M}_k\left[\gamma_k \mathbfit{s}_k(n)-\gamma_k \mathbfit{M}_k^{-1} \boldsymbol{\lambda}_k(n)\right] &\leq 0\nonumber
\end{align}
from which inequality \eqref{ieq_bap} follows.
\end{proof}

The first part of the proof for Theorem \ref{thm:tsit} closely follows the proof in \cite[Theorem 5.6.1]{2140467}. Therefore, we only present the basic idea and the parts which differ from the original proof. 
Starting with the second-order Taylor expansion of $U$, we derive\footnote{Let $f:\mathbb{R}^N\rightarrow\mathbb{R}$ be a continuously differentiable function. Based on Taylor's remainder theorem, we have $\forall \mathbfit{x},\mathbfit{s}\in\mathbb{R}^N,\exists \mathbfit{y}\in [\mathbfit{x},\mathbfit{x}+\mathbfit{s}]$ such that $f(\mathbfit{x}+\mathbfit{s})=f(\mathbfit{x}) + \mathbfit{s}^T\nabla f(\mathbfit{x}) + \frac{1}{2}\mathbfit{s}^T\nabla^2 f(\mathbfit{y})\mathbfit{s}$. Further, the quadratic term is lower bounded by $\mathbfit{s}^T\nabla^2 f(\mathbfit{y})\mathbfit{s} \geq -\sum_k \left|s_k\right|^2\sum_l \left|\frac{\partial^2 f(\mathbfit{y})}{\partial x_k\partial x_l}\right|$.} a lower of bound for $U(\mathbfit{X}(n+1))$ as 
\begin{align}
U(\mathbfit{X}(n+1))&\geq U(\mathbfit{X}(n)) + \sum_{k=1}^K \gamma_k \nabla_k U(\mathbfit{X}(n))\mathbfit{s}_k(n) \nonumber\\
&\quad - \frac{1}{2}\sum_{k=1}^K\sum_{l=1}^K \gamma_k^2  \left|s_{k,l}(n)\right|^2 \sum_{s=1}^K K_{kl,s}. \label{PXeq:taylorLB}
\end{align}
By the Block Ascent Property (Lemma \ref{lemma2}) we have
\begin{align}
\nabla_k U(\mathbfit{X}(n))\mathbfit{s}_k(n)&= \boldsymbol{\lambda}_k(n)^T \mathbfit{s}_k(n) + \nonumber\\
& \quad \; [\nabla_k U(\mathbfit{X}(n)) - \boldsymbol{\lambda}_k(n)^T]\mathbfit{s}_k(n) \nonumber\\
&\geq \sum_{l=1}^K\beta_{k,l} \left|s_{k,l}(n)\right|^2 + \nonumber \\
& \quad \; [\nabla_k U(\mathbfit{X}(n)) - \boldsymbol{\lambda}_k(n)^T]\mathbfit{s}_k(n).
\end{align}
After some algebraic manipulations and summing for different values of $n$, we obtain
\begin{align}
U(\mathbfit{X}(n+1))&\geq U(\mathbfit{X}(0)) \nonumber\\
&+ \sum_{p=0}^n\sum_{k=1}^K\sum_{l=1}^K \frac{1}{2} \gamma_k \left|s_{k,l}(p)\right|^2 \left(2 \beta_{k,l} -\gamma_k D_{k,l}\right)
\end{align}
with $D_{k,l}=\sum_{s=1}^K K_{kl,s}(1+P_{s,l}+Q_{l,k})+K_{sl,k}\left(P_{k,l}+Q_{l,s}\right)$.\\
Let $G_k=\mathrm{min}_l\, \slfrac{2\beta_{k,l}}{D_{k,l}}$ and assume that $\gamma_k\in(0,G_k)$. Then we have some $C_1>0$ with
\begin{align}
0< C_1 \leq  2 \beta_{k,l}-\gamma_k D_{k,l},\forall k,l\nonumber
\end{align}
for which it holds 
\begin{align}
 \sum_{k=1}^K\sum_{l=1}^K \sum_{p=0}^n \gamma_k \left|s_{k,l}(p)\right|^2 \leq \frac{2}{C_1}\left[U(\mathbfit{X}(n+1)) - U(\mathbfit{X}(0))\right].
\end{align}
Since $U$ is bounded from above and $U\geq 0$, we have a finite $C_2\geq U(\mathbfit{X}(n+1)) - U(\mathbfit{X}(0))$ such that for every $k,l$ and $p\geq 0$,
\begin{align}
\sum_{p=0}^\infty  \left|s_{k,l}(p)\right|^2 \leq  \frac{2C_2}{\gamma_k C_1}< \infty,
\end{align}
which implies $\lim_{p\rightarrow\infty}  \left|s_{k,l}(p)\right|^2 =0$. By \eqref{eq_x_s} we have $\lim_{p\rightarrow\infty} \mathbfit{x}_k(p+1)-\mathbfit{x}_k(p) =0$; that is, for every $k$ the sequence of power gain vectors $\left\{\mathbfit{x}_k(p)\right\}$ converges to a limit point $\mathbfit{x}_k^*$.\\
Finally, we show that the limit point satisfies the Optimality Condition \eqref{eq:optcond1}. Let $T_k^p:\Omega_k \rightarrow \Omega_k$ be the mapping that corresponds to the $p$-th iteration of the asynchronous SGP algorithm (i.e., $\mathbfit{x}_k(p+1)=T_k^p(\mathbfit{x}_k(p))$), then by \cite[Proposition 3.7 (e)]{Bertsekas:1989:PDC} we have $T_k^p(\mathbfit{x}_k^*) = \mathbfit{x}_k^*$ if and only if $\nabla_k U\left(\mathbfit{X}^*\right)(\mathbfit{y}-\mathbfit{x}_k^*)\leq 0,\forall \mathbfit{y}\in\Omega_k$, which completes the proof.

\section{Example: Proportional Fair Rate}\label{ExamplePF} 
We now illustrate the calculation of the curvature bounds for the proportional fair rate utility \eqref{eq:PFutil}
with $u_k(\mathbfit{x}^k) = \log_2\left(1 + \Gamma_k(\mathbfit{x}^k)\right)$ and $\Gamma_k(\mathbfit{x}^k)=x_{k,k}/(\sum_{l\neq k} x_{l,k} + \sigma^2)$. An equivalent problem formulation, which satisfies Assumption \ref{Uprop}, is obtained by transforming the objective $U^\mathrm{pf}$ with the monotonously increasing $\log(.)$ function (cf. \cite{Boyd}), yielding
\begin{align}
\arg\max U^\mathrm{pf} &\equiv \arg\max \sum_k \log u_k + c_k.
\end{align}
The constants $c_k,\forall k$ can be chosen such that Assumption \ref{ukProp} is satisfied. However, they do not depend on $\mathbfit{X}$ and thus can be omitted.\\
Next, let $U=\sum_k \log u_k$ and  $P_k^\mathrm{rx}=\sum_{l=1}^K x_{l,k} + \sigma^2$. The second-order partial derivatives of $U$ are given by
\begin{align}
\frac{\partial^2 U}{\partial x_{k,k}^2}&=-\frac{1}{(P_k^\mathrm{rx})^2}\left[\frac{1}{u_k^2}+\frac{1}{u_k}\right], \label{P2eq:Kkkk}\\
\frac{\partial^2 U}{\partial x_{k,k}\partial x_{l,k}}&=\frac{\partial^2 U}{\partial x_{l,k}\partial x_{k,k}}=\frac{1}{(P_k^\mathrm{rx})^2}\left[\frac{\Gamma_k}{u_k^2}-\frac{1}{u_k}\right], \label{P2eq:Kkkl}\\
\frac{\partial^2 U}{\partial x_{l,k}^2}&=\frac{\partial^2 U}{\partial x_{l,k} \partial x_{s,k}}=\frac{1}{u_k}\frac{\Gamma_k^2}{(P_k^\mathrm{rx})^2}\left[\frac{2}{\Gamma_k}+1-\frac{1}{u_k}\right]. \label{P2eq:Kskl}
\end{align}

A set of valid (global) curvature bounds $K_{lk,s},\forall l,s$ can be determined by the $k$-th user solely on the basis of its local CSI knowledge; that is, the $k$-th user needs to determine for each $l,s$ 
\begin{align}
K_{lk,s}=\max_{(x_{1,k},\ldots,x_{K,k})\in \mathcal{I}_{1}\times \ldots \times \mathcal{I}_{K}} \left|\frac{\partial^2 U}{\partial x_{l,k}\partial x_{s,k}}\right|
\end{align}
with the intervals $\mathcal{I}_{l}=[0,\left\|\mathbfit{h}_{lk}\right\|^2],l\neq k$ and  $\mathcal{I}_{k}=[\frac{1}{\mu},\left\|\mathbfit{h}_{kk}\right\|^2],\mu\in\mathbb{R}_{++}$.

Finally, we show that the transformed problem formulation satisfies Assumption \ref{P2monCurvAssm}.
First note that the second-order partial derivatives in \eqref{P2eq:Kkkk}, \eqref{P2eq:Kkkl}, \eqref{P2eq:Kskl} do not change their sign\footnote{By omitting the strictly positive factor $\slfrac{1}{(P_k^\mathrm{rx})^2}$, we can treat the second-order partial derivative $\slfrac{\partial^2 U}{\partial x_{l,k}\partial x_{s,k}}$ as a function $f_{lk,s}(\Gamma_k)$ that solely depends on $\Gamma_k$. This function can be analyzed in order to derive the statement.} for all $\mathbfit{x}^k\in\mathbb{R}_+^K$ with $x_{k,k}>0$. Thus, we can neglect the absolute value operator for their monotonicity analysis. By inspecting the partial derivatives of the expressions \eqref{P2eq:Kkkk}, \eqref{P2eq:Kkkl}, \eqref{P2eq:Kskl} with respect to $x_{l,k},\forall l$, we observe the following relationships:
\begin{itemize}
	\item \eqref{P2eq:Kkkk} is monotonically decreasing with respect to $x_{k,k}$, and monotonically increasing with respect to $x_{l,k},\forall l\neq k$.
	\item \eqref{P2eq:Kkkl} and \eqref{P2eq:Kskl} are monotonically decreasing with respect to $x_{l,k},\forall l$.
\end{itemize}
Thus, the utility $U=\sum_k \log u_k$ satisfies Assumption \ref{P2monCurvAssm}.


%
%

\ifCLASSOPTIONcaptionsoff
  \newpage
\fi



%

\bibliographystyle{IEEEtran}
\bibliography{bibtex_refs_TSP}

%
%

%

%
%
%




\end{document}